\documentclass[preprint,12pt]{elsarticle}

\usepackage{graphicx}

\usepackage{url}

\usepackage{amsmath,amssymb}
\usepackage{amsthm}
\usepackage{color}
\usepackage{multirow}
\usepackage{rotating}
\usepackage{colortbl}
\usepackage[caption=false,font=footnotesize]{subfig}

\usepackage{tikz}
\usetikzlibrary{shapes,arrows,decorations}

\journal{Science of Computer Programming}

\newcommand{\goes}[1]{\xrightarrow[]{#1}}

\newtheorem{thm}{Theorem}
\newtheorem{cor}{Corollary}
\newtheorem{lem}{Lemma}
\newtheorem{prop}{Proposition}
\newtheorem{defn}{Definition}

\newtheorem{example}{\ \\Example}

\definecolor{lGray}{rgb}{0.9,0.9,0.9}

\def\name#1{\mbox{\sc #1}}

\def\sos#1#2{{\def\arraystretch{1.6}\begin{array}{c}#1\\\hline
#2\end{array}}}

\begin{document}

\begin{frontmatter}



\title{Adaptability Checking in Multi-Level Complex Systems}

\author{Emanuela Merelli}
\ead{emanuela.merelli@unicam.it}
\author{Nicola Paoletti}
\ead{nicola.paoletti@unicam.it}
\author{Luca Tesei\corref{cor1}}
\cortext[cor1]{Corresponding author}
\ead{luca.tesei@unicam.it}
\address{School of Science and Technology, Computer Science Division, University of Camerino, Via del Bastione 1, 62032, Camerino, Italy}

\begin{abstract}
A hierarchical model for multi-level adaptive systems is built on two basic levels: a lower behavioural level $B$ accounting for the actual behaviour of the system and an upper structural level $S$ describing the adaptation dynamics of the system. The behavioural level is modelled as a state machine and the structural level as a higher-order system whose states have associated logical formulas (constraints) over observables of the behavioural level. $S$ is used to capture the global and stable features of $B$, by a defining set of allowed behaviours. The adaptation semantics is such that the upper $S$ level imposes constraints on the lower $B$ level, which has to adapt whenever it no longer can satisfy them. In this context, we introduce weak and strong adaptability, i.e.\ the ability of a system to adapt for some evolution paths or for all possible evolutions, respectively. We provide a relational characterisation for these two notions and we show that adaptability checking, i.e.\ deciding if a system is weak or strong adaptable, can be reduced to a CTL model checking problem. We apply the model and the theoretical results to the case study of motion control of autonomous transport vehicles.
\end{abstract}


\end{frontmatter}



\section{Introduction}
\label{sect:intro}
Self-adaptive systems are particular systems able to modify their own behaviour according to their environment and their current configuration. They learn and develop new strategies in order to fulfil an objective, to better respond to problems, or, more generally, to maintain desired conditions. 

From a broad viewpoint, self-adaptiveness is an intrinsic property of complex natural systems. Self-adaptation is a process driving both the evolution and the development of living organisms that adapt their features and change their phenotype in order to survive to the current habitat, to achieve higher levels of fitness and to appropriately react to external stimuli.

Nowadays, software systems are increasingly resembling complex systems, which motivates the development of methods for enabling software self-adaptiveness. Similarly to natural systems, \textit{``Self-adaptive software evaluates its own behaviour and changes behaviour when the evaluation indicates that it is not accomplishing what the software is intended to do, or when better functionality or performance is possible.''}~\cite{laddaga}. Self-adaptive software finds application in fields like autonomic computing, service-oriented architectures, pervasive service ecosystems, mobile networks, multi-agent systems, and ultra-large-scale (ULS) software systems~\cite{feiler2006ultra}, characterised by distributed, autonomous, interacting, heterogeneous, conflicting and evolvable sub-systems.

%
%


\subsection{Contributions}
In this work we develop a formal hierarchical model for multi-level self-adaptive systems, where two fundamental levels are defined: the \textit{lower behavioural level} $B$, which describes the admissible dynamics of the system; and the \textit{upper structural level} $S$, accounting for the global and stable features of the system that regulates the lower behaviour. More precisely, the $B$ level is modelled as a state machine and the $S$ level is also modelled as a state machine, but such that each state is associated with a set of \textit{constraints} (logical formulas) over observable variables of the lower level. 

A state in the structural level (also called $S$ state) represents a relatively persistent situation, a steady region of the $B$ level, identified by the set of $B$ states satisfying the constraints. Therefore the $S$ level underlies a \textit{higher order structure}, because $S$ states can be interpreted as sets of $B$ states and, consequently, $S$ transitions relate sets of $B$ states. 

In the remainder of the paper our model will be referred to as $S[B]$, in order to highlight the two basic levels that compose the system. This model is broadly inspired by a spatial bio-inspired process algebra called \textit{Shape Calculus}~\cite{shape2,shape1}, where processes are characterised by a reactive behaviour $B$ and by a shape $S$ that imposes a set of geometrical constraints on their interactions and occupancy in the three-dimensional Euclidean space. Here, instead, this paradigm is shifted in a more general context, where $S$ and $B$ are entangled by a hierarchical relation defined on the structural constraints of the $S$ level and the state space of the $B$ level. 

In the following, a brief description of the adaptation semantics is given. Let $\bar{q}$ be the current $B$ state and $\bar{r}$ be the current $S$ state of an $S[B]$ system. Adaptation is triggered whenever $\bar{q}$ cannot evolve into a new state satisfying the current constraints specified by $\bar{r}$. At this point, the $B$ level attempts to adapt towards a target $S$ state $r'$, reachable by performing a transition from $\bar{r}$. During adaptation $B$ is no more constrained by the $S$ level, apart from an invariant condition (possibly empty) guaranteeing that some sanity conditions are met during this phase. Such an invariant is defined locally to the adaptation phase from $r$ to the target $r'$. Adaptation terminates successfully when the $B$ level ends up in a state $q'$ that fulfils the constraints of (one of) the target(s) $r'$.

After the definition of a model for adaptive systems and of a particular adaptation semantics over it, we focus on the \textit{adaptability checking problem}, i.e.\ checking if the system is able to adapt successfully for:
\begin{itemize}
\item some adaptation paths (\textit{weak adaptability checking}); or
\item all possible adaptation paths (\textit{strong adaptability checking}).
\end{itemize}

To this purpose, we set up a formal framework (see Sect.~\ref{sect:adaptability}), based on the definition of weak and strong adaptation as relations over the set of $B$ states and the set of $S$ states. In this way, adaptability is verified on an $S[B]$ system if an appropriate adaptation relation can be built over the states of $B$ and the states of $S$. We formulate the notions of weak and strong adaptability also in a logical form, as Computation Tree Logic (CTL) formulas over the given semantics of an $S[B]$ system. Then, by proving that the logical characterization is equivalent to the relational one, we demonstrate that the adaptability checking problem can be reduced to a classical model checking problem.

A first general introduction of the $S[B]$ model was given in~\cite{merelli2012} by the same authors. In this work, we provide several novelties and improvements, most of them devoted to the adaptability checking problem. Here we show the effectiveness of $S[B]$ systems on a case study in the context of adaptive software systems, a motion controller model for autonomous transport vehicles (ATVs). In~\ref{appendix:bone}, our approach is validated also with a model of bone remodelling, a biological system that is intrinsically self-adaptive, thus showing how $S[B]$ systems are potentially suitable to analyse a broad class of adaptive systems. In addition, we update the operational semantics of $S[B]$ systems and we prove several properties of the resulting labelled transition systems and of the associated weak and strong adaptability relations. Regarding the adaptability checking problem, we formally prove the equivalence between the relational and the logical formulation of strong and weak adaptability (Theorems~\ref{teo:weak} and~\ref{teo:strong}). Finally, we discuss the computational complexity of the adaptability checking problem.


\subsection{Adaptation features of $S[B]$ systems}
\label{sec:introadapt}
Let us characterise our approach according to the ``taxonomy of self-adaptation'', a quite general software-oriented classification proposed by Salehie and Tahvildari~\cite{salehie2009self}. Specifically, the features of the taxonomy considered here are \textit{adaptation type}, or \textit{how} adaptation is realized; temporal issues, related to \textit{when} the system needs to change and to be monitored to achieve adaptation; and \textit{interaction} aspects.

\noindent \textbf{Interaction.} In the $S[B]$ model communication and interactions of the adaptive system with other systems are not explicitly taken into account. This is because the main purpose of the current work is studying the adaptation capabilities of a fundamental model of computation, to which more powerful and expressive models can typically be reduced. Indeed, we always consider the behavioural level $B$ as the transitional semantics of a system constructed from several interacting components.

\noindent \textbf{Temporal characteristics.} Recalling the introductory description of the adaptation semantics, \textit{adaptation starts as late as possible}, only when no other evolution is possible that fulfils the current constraints; and \textit{adaptation ends as soon as possible}, i.e.\ as soon as a target state can be reached. This implies that $S[B]$ systems support a basic type of \textit{proactive} (i.e.\ anticipatory) adaptation, which ensures that the system reaches a state where structural constraints or adaptation invariants are not met if and only if no other evolution is possible.

The choice to exclude adaptations starting from states that can progress normally, i.e.\ without violating the constraints, is motivated by the same definition of adaptation: a mutation in an individual that leads to a \textit{higher} level of fitness. Indeed, as stated in~\cite{bouchachia2012introduction}, an adaptive system \textit{``\ldots seeks to configure its structure with the overall aim of adaptation to the environment trying to optimize its function (i.e., to maximize its fit) to meet its reason of existence''}. 
Our model provides just a qualitative characterization of the fitness of a $B$ state $\bar{q}$ in a $S$ state $\bar{r}$, given by the satisfaction value (true or false) of the constraints. Therefore, no adaptations can start from the state $\bar{q}$ if it can satisfy the constraints in $\bar{r}$ during its evolution, since this configuration corresponds to the highest possible fitness, and any adaptation would produce \textit{equal} (in case of successful adaptation) or \textit{lower} (unsuccessful adaptation) fitness values.

\noindent \textbf{Adaptation type.} This feature concerns aspects related to the implementation of adaptation mechanisms. According to the taxonomy above, our approach falls into the definitions of \textit{model-based adaptation}, i.e.\ based on a model of the system and of the environment; and of \textit{close adaptation}, in the sense that the system has only a fixed number of adaptive actions, due to the fact that we focus on models with finite and fixed state space. On the contrary, \textit{open-adaptive} systems support the runtime addition of adaptation actions. Salehie and Tahvildari also distinguish between \textit{external} and \textit{internal} adaptation. $S[B]$ systems belong to the former type, where adaptation logic ($S$ level) and application logic ($B$ level) are separated. In internal adaptation, conversely, adaptation mechanisms are mixed at the application level.

Taking a broader view that generalizes from software systems, Sagasti~\cite{sagasti1970conceptual} distinguishes between two different adaptive behaviours: the system adapts by modifying itself (\textit{Darwinian adaptation}); or it adapts by modifying its environment (\textit{Singerian adaptation}). In this work, we clearly focus on the former type of adaptation. 
Following this line, the adaptation type of $S[B]$ systems can be further classified as a \textit{top-down} and \textit{behavioural} adaptation. Top-down, because the $S$ level imposes high-level functions (e.g.\ constraints, rules and policies) on the lower $B$ level, which adapts itself whenever it cannot fulfil the current constraints. Bottom-up adaptation represents the opposite direction, occurring for instance when new higher-level patterns emerge from the lower level. On the other hand, behavioural adaptation is related to functional changes, like changing the program code or following different trajectories in the state space. In literature, it is generally opposed to structural adaptation, which is related to architectural reconfiguration, e.g.\ addition, migration and removal of components. Note that structural and behavioural adaptation must not be confused with the structural and the behavioural level of an $S[B]$ system.

\subsubsection*{Structure of the paper}
The paper is organized as follows. Section~\ref{sect:model} introduces the
formalism and the syntax of the $S[B]$ model. Section~\ref{subsect:exmpl_atv} illustrates the application of the model to the example of adaptive motion control. In Section~\ref{sect:semantics} we give
the operational semantics of an $S[B]$ system by means of a flattened transition
system. In Section~\ref{sect:adaptability} we formalise the relations of weak and
strong adaptation, which we equivalently characterise as CTL formulas in Section~\ref{sec:logic}. Related works and conclusions are given in 
Section~\ref{sect:conclusion}. Finally, proofs are presented in~\ref{appendix:proofs} and the $S[B]$ model of bone remodelling in~\ref{appendix:bone}.
\section{A Formal Hierarchical Model for Adaptive Systems}
\label{sect:model}
In our model, a system encapsulates both the behavioural ($B$) and the
structural/adaptive ($S$) level. The behavioural level is classically described as a finite state machine of the form $B=(Q, q_0, \goes{}_B)$ ($Q$ set of $B$ states, $r_0$ initial $B$ state and $\goes{}_B$ transition relation). The structural level is modelled as a finite state machine $S = (R, r_0, \mathcal{O}, \goes{}_S, L)$  where $R$ is a set of $S$ states, $r_0$ is the initial $S$ state, 
$\mathcal{O}$ is an observation function, $\goes{}_S$ is a
transition relation and $L$ is a state labelling function. The function $L$ labels each $S$ state with a formula representing a set of constraints over an \emph{observation} of the $B$ states. 
Therefore an $S$ state $r$ can be directly mapped to the set of
$B$ states satisfying $L(r)$. Through this hierarchy,  $S$ can be viewed as a {\em second-order} structure $(R \subseteq 2^Q, r_0, \goes{}_S \subseteq 2^Q \times 2^Q, L)$ where each $S$ state $r$ is identified with its corresponding set of $B$ states.

Behavioural adaptation is achieved by switching from an $S$ state
imposing a set of constraints to another $S$ state where a (possibly) different
set of constraints holds. During adaptation the behavioural level is no more
regulated by the structural level, except for a condition, called
\textit{transition invariant}, that must be fulfilled by the system undergoing
adaptation. We can think of this condition as a minimum requirement to which the
system must comply when it is adapting and, thus, it is not
constrained by any $S$ state. Note that specifying $true$ as transition invariant one can allow the system to adapt by followinSg any trajectory that is present at the $B$ level. 


\subsection{Information processing between S and B}
\label{sec:information}

In general, we assume no reciprocal internal knowledge between the $S$  and the
$B$ level. In other words, they see each other as \textit{black-box systems}.
However, in order to realise our notion of adaptiveness, there must be some information processing both from $B$ to $S$ and from $S$ to $B$. In particular,
the information from the $B$ level processed by the $S$ level is modelled here as a set of variables $A = \{a_1, \ldots,
a_n\}$ called \emph{observables} of the $S$ level on the $B$ level. 
The values of these variables must always be \emph{derivable} from the
information contained in the $B$ states, which can possibly hold more ``hidden''
information related to unknown interactions and internal activity.
This keeps our approach
black-box-oriented because the $S$ level has not the full knowledge of the
$B$ level, but only some derived (e.g.\ aggregated, selected or calculated)
information.

The adaptation model of an $S[B]$ system could be viewed as a closed-loop system, illustrated in Fig.~\ref{fig:control-loop}, where, in control-theoretic terms, the $B$ level would represent the plant, and the $S$ level the controller.

Let $\bar{q}$ and $\bar{r}$ denote the current state of $B$ and $S$, respectively. $B$ outputs the vector $\mathbf{x}=Post(\bar{q})$~\footnote{With abuse of notation, we allow the $Post$ operator to return an indexed vector of states instead of a set.} of the states reachable from $\bar{q}$ with a single transition. An element $x_i$ of $\mathbf{x}$ would be such that $\bar{q} \goes{} x_i$ and of course we exclude replicated states: $\bigwedge_{i\neq j} x_i \neq x_j$.
Since we assume that $S$ cannot directly access to $B$ states but only to the values of the variables observed at those states, an observer feeds $S$ with the observations $\mathbf{o} = \mathcal{O}(x_i)$ at each next state.

According to the operation mode $m$ (steady or adapting), to the inputs from $B$, to its current state $\bar{r}$ and to the (possibly null) target $S$ state $r$, $S$ updates the current operation mode, its current and target state and computes a vector $\mathbf{v}$ for selecting the allowed next states of $B$. In particular, an element $v_i$ of $\mathbf{v}$ is true iff, under the observation $\mathbf{o}$, either the invariant of some adaptation paths or the current constraints in $\bar{r}$ are satisfied, depending on whether the system is adapting or not.

The $S$ level closes the feedback loop by outputting $\mathbf{v}$ to $B$, which in turn can update its current state $\bar{q}$ by selecting one of the allowed states, i.e.\ those next states $x_i$ under which the required constraints are met ($v_i$ true). Such set can be expressed in function of $\mathbf{v}$ and $\mathbf{x}$ as $f(\mathbf{v},\mathbf{x})=\{x_i \in \mathbf{x} \ | \ v_i = \top\}$.


\tikzstyle{block} = [draw, rectangle, very thick,draw=black!50,
    top color=white,bottom color=black!20,
    minimum height=5em, rounded corners=3mm]
\begin{figure}
\centering
\begin{tikzpicture}[auto,>=latex',font=\footnotesize]
   \node [block] (bsystem) {$\begin{array}{c} B\text{ level}\\ \\ { \bar{q}\in f(\mathbf{v},\mathbf{x})} \\ \mathbf{x}\leftarrow Post(\bar{q})\\ \end{array}$};
      
   \node [block, below of=bsystem, node distance=2.5cm] (ssystem){$\begin{array}{c} S\text{ level}\\ \\
(m,\bar{r}, r, \mathbf{v})\leftarrow g(m,\bar{r}, r, \mathbf{o})
    \end{array}$};
       \node [block, name=observer, right of=ssystem, minimum height=2em,node distance=3.5cm] {$\mathcal{O}$};
          \node [coordinate,name=inputb, left of=ssystem,node distance=3.5cm] {};
          
	\draw [->,rounded corners=2mm,very thick,color=black!50] (bsystem) -| node [color=black]{$\mathbf{x}$} (observer);
	 \draw [->,rounded corners=2mm,very thick,color=black!50] (observer) -- node [color=black]{$\mathbf{o}$} (ssystem);
	\draw [-,very thick,draw=black!50] (ssystem)-- (inputb);
	\draw [->,rounded corners=2mm,very thick,draw=black!50] (inputb) |- node {$\mathbf{v}$} (bsystem);

\end{tikzpicture}
\caption{Adaptation loop in an $S[B]$ system. At each step, the $S$ level observes the next states $\mathbf{x}$ of the $B$ level and closes the loop by outputting to $B$ a vector $\mathbf{v}$ which describes the allowed next states, e.g.\ those that satisfy the current constraints if the system is not adapting, or those that satisfy the current adaptation invariant.}
\label{fig:control-loop}
\end{figure}
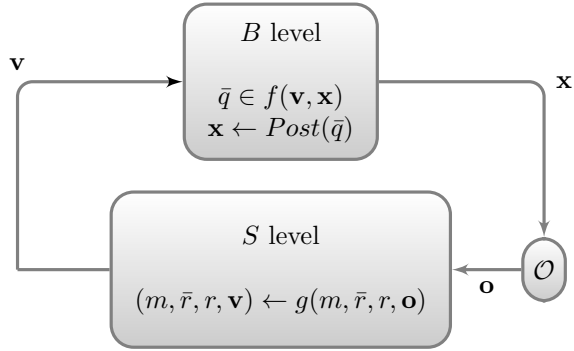

%
\subsection{Language for constraints}

In our model, the constraints characterising the states of an $S$ level are expressed using formulas of a many-sorted first order logic. More precisely, the definition of an $S$ level includes the definition of a many-sorted signature $\Sigma$ containing some function symbols, some predicate symbols and some sorts $D_{1}, \ldots, D_{k}$. $\Sigma$-terms and $\Sigma$-formulas are constructed in the standard way \cite{feferman1974applications}. 
In addition, a particular set of sorted variables, which we call observables, must be fixed in the $S$ level. Such a set is of the form $A = \{ a_{1} \colon D_{j_{1}}, \ldots, a_{n} \colon D_{j_{n}}\}$, where $j_{i} \in \{1, \ldots, k\}$ for all $i=1, \ldots, n$. Then, constraints can be expressed as $\Sigma$-formulas $\psi$ such that the variables that occur free in $\psi$, denoted by $\mathrm{free}(\psi)$, are a (possibly empty) subset of $A$. This set will be denoted by $\Psi(\Sigma, A) = \{\psi \mid \psi \mbox{ is a $\Sigma$-formula} \wedge \mathrm{free}(\psi) \subseteq A \}$.
 
We also impose that in the $S$ level a particular structure $M$ is fixed for the evaluation of $\Sigma$-formulas, i.e.\ $k$ non-empty domains $M(D_{1}), \ldots, M(D_{k})$, as carrier sets for sorts, together with interpretations for all function and predicate symbols of $\Sigma$. To obtain the full semantic evaluation of formulas in $\Psi(\Sigma, A)$ we will take values for the free variables in $A$ from an observation function. 

\begin{defn}[Observation Function]
Let $\mathcal{Q}$ be the universe set of all states of machines possibly representing $B$ levels. Let $\Sigma$ be a many-sorted signature, let 
$A=\{ a_{1} \colon D_{j_{1}}, \ldots, a_{n} \colon D_{j_{n}}\}$ 
be a set of observables and let $M$ be a structure for the evaluation of $\Sigma$-formulas. 
An \emph{observation function} $\mathcal{O}^{\Sigma,A}_{M}$ on $\Sigma$, $A$ and $M$ is a partial function 
$$
\mathcal{O}^{\Sigma,A}_{M} \colon \mathcal{Q} \hookrightarrow (A \rightarrow \mathcal{D})
$$
where $(i)$ $\mathcal{D} = \bigcup_{i=1}^{n} M(D_{j_{i}})$ 
and 
$(ii)$ for any state $q \in \mathcal{Q}$, if $\mathcal{O}^{\Sigma,A}_{M}(q) \neq \perp$ then $\mathcal{O}^{\Sigma,A}_{M}(q)(a_{i} \colon D_{j_{i}}) \in M(D_{j_{i}})$, for all $i = 1, \ldots, n$. For a lighter notation, we will use $\mathcal{O}$ instead of $\mathcal{O}^{\Sigma,A}_{M}$ when $\Sigma, A$ and $M$ are clear from the context.
\end{defn}

Note that the use of the universe of states as domain makes the definition of the observation function independent from a particular state machine representing a behavioural level $B$. Note also that we do not require the observation function to be injective. This means that some different states can give the same values to the observables. In this case, the difference among the states is not visible to $S$ through the observation, but it is internal to $B$. 

To complete the machinery for checking whether a set of constraints is satisfied or not, we define the satisfaction relation in the natural way.

\begin{defn}[Satisfaction relation]
Let $\mathcal{O}^{\Sigma,A}_{M}$ be an observation function. A state $q \in \mathcal{Q}$
\emph{satisfies} a formula $\psi \in \Psi(\Sigma,A)$, written $q \models \psi$,
iff $\mathcal{O}^{\Sigma,A}_{M}(q) \neq \perp$ and $\psi$ is true, according to the standard semantics of many-sorted first order logic, with respect to the structure $M$ and by substituting in $\psi$ every occurrence of the free sorted variable $a_{i} \colon D_{j_{i}}$ with $\mathcal{O}^{\Sigma,A}_{M}(q)(a_{i} \colon D_{j_{i}})$.

We also define an evaluation function $[[ \cdot ]]:
\Psi(\Sigma,A) \goes{} 2^\mathcal{Q}$ mapping a formula
$\psi \in \Psi(\Sigma,A)$ to the set of states $[[ \psi ]] = \lbrace q \in \mathcal{Q} \mid q
\models \psi \rbrace$, i.e.\ those satisfying $\psi$.
\end{defn}

We can now state that what we call a set of constraints is formally expressed by a formula $\psi \in \Psi(\Sigma,A)$ that is the conjunction of all the formulas representing each constraint in the set. The set of constraints is satisfied if and only if the corresponding formula is true in the fixed structure $M$ and observation $\mathcal{O}^{\Sigma,A}_{M}$. 


\begin{example}
Let us consider a set of observables and associated sorts:
$$
A = \{{velocity}:\mathbb{R}, \ {congestion}:\mathbb{B} \}
$$ 
Consider also a signature $\Sigma = \{ \mathbb{R}, \mathbb{B}, ==,>,<,0,5\}$ where $\mathbb{R}$ and $\mathbb{B}$ are the sorts indicating the domains of real numbers and boolean, respectively; $==$ is the equality predicate interpreted as the identity relation in each domain; $>$ and $<$ are the usual greater-than and less-than predicates over $\mathbb{R}$; and the constants $0$ and $5$ are the real numbers $0$ and $5$. A possible formula $\psi$ in the language $\Psi(\Sigma, A)$ is 
$$
congestion \implies velocity < 5 \ \wedge \ \neg congestion \implies velocity > 0
$$
whose satisfaction depends on the particular values of the variables, which will be different in different states.

In the context of Autonomous Transport Vehicles (ATVs), this formula can be thought to represent a set of two constraints, one imposing that ``in case of congestion, the velocity of the vehicle must be lower than $5$'' and the other that ``in normal traffic conditions, the velocity must be greater than $0$''.
\end{example}

\subsection{Coupling $S$ and $B$}

Let us now formally define the behavioural level $B$ and the structural level $S$ separately. Afterwards, the $S[B]$ model is defined as the combination of the two.

\begin{defn}[Behavioural level]
The behavioural level of a system is a tuple $B=(Q, q_0, \goes{}_B)$, where
\begin{itemize}
\item $Q \subseteq \mathcal{Q}$ is a \emph{finite} set of states and $q_0 \in Q$ is the initial state; and
\item $\goes{}_B \subseteq Q \times Q$ is the transition relation.
\end{itemize}
\end{defn}


\begin{defn}[Structural Level]
The structural level of a system is a tuple $S=(R, r_0, \mathcal{O}^{\Sigma,A}_{M}, \goes{}_S, L)$,
where
\begin{itemize}
\item $R$ is a finite set of states and $r_0 \in R$ is the initial state;
\item $\mathcal{O}^{\Sigma,A}_{M}$ is an observation function on a signature $\Sigma$,  a set of observables $A$ and a structure $M$;
\item $\goes{}_S \subseteq R \times \Psi(\Sigma,A) \times R$ is a transition relation,
labelled with a formula called \emph{invariant}; and
\item $L: R \goes{} \Psi(\Sigma,A)$ is a function labelling each state with a formula representing \emph{a set of constraints}.
\end{itemize}
\end{defn}

%

Let us now give an intuition of the adaptation semantics. Let the current
$S$ state be $r_i$ and suppose $r_i \goes{\psi}_S r_j$ for some $r_{j}$. Assume that the behaviour is
in a steady state (i.e.\ not adapting) $q_i$ and therefore $q_i \models L(r_i)$.
If the $B$ state can move, but all $B$ transitions $q_i \goes{}_B q_j$ are such that $q_j \not  \models
L(r_i)$, then the system can start adapting to the target $S$ state $r_j$. In this phase, the
$B$ level is no more constrained, but during adaptation the invariant $\psi$
must be met. Adaptation ends when the behaviour reaches a state $q_k$ such that
$q_k \models L(r_j)$. 

We want to remark that the model supports the non-deterministic choice between adaptations, i.e.\ the system can adapt to every target state $r_{j}$ reachable with a transition $r_i \goes{\psi}_S r_j$ from the current $r_{i}$ state. The non-determinism can be both external - that is different target states can be reached by satisfying different invariants - and internal - that is different target states can be reached satisfying the same invariant condition. 

\begin{defn}[\mbox{$S[B]$ system}]\label{def:sb} \ \\
An \emph{$S[B]$ system} is the combination of a behavioural level $B=(Q, q_0, \goes{}_B)$ and a structural level $S=(R, r_0, \mathcal{O}^{\Sigma,A}_{M}, \goes{}_S, L)$ such that for all $q \in Q$, $\mathcal{O}^{\Sigma,A}_{M}(q) \neq \perp$. Moreover, in any $S[B]$ system the initial $B$ state must satisfy the constraints of the initial $S$ state, i.e.\ $q_0 \models L(r_0)$.
\end{defn}
\section{Case Study: Adaptive Motion Control of Autonomous Transport Vehicles}\label{subsect:exmpl_atv}
In this section, we illustrate the features of our approach by means of an example adapted from~\cite{khakpour2012hpobsam}: a model of motion control of \textit{Autonomous Transport Vehicles (ATVs)} in a smart airport. In this context, ATVs are responsible for the transport of passengers between stopovers like passenger entrances, check-in desks, departure gates, and plane parking positions. In this work, we just consider a subcomponent of such vehicles, that accounts for controlling vehicle speed and for switching from main roads to secondary roads in case of traffic congestion or blockages.

\subsubsection*{Behavioural Level}
The behavioural level model is depicted in Fig.~\ref{fig:b-level_atv} and over it, we consider the following set of observable variables and associated sorts:
\begin{itemize}
\item $r: \{M\text{ (main)}, S\text{ (secondary)}\}$, the road being driven on;
\item $v: \{V_0\text{ (slow)}, V_1\text{ (medium)},V_2\text{ (high)}\}$, the velocity of the vehicle; and 
\item $c: \{0 \text{ (no congestion)},1 \text{ (congestion)} \}$, a variable indicating the congestion of the main road.
\end{itemize}
Hereafter, we will refer to a specific state $s$ by using the notation $s$:$(r_s,v_s,c_s)$, where $r_s$, $v_s$ and $c_s$ are the values of the observables at $s$.

At each state, the motion controller can increase, decrease or keep the current velocity; and it can switch from the main road to the secondary one, or viceversa. In this example, we simulate that a congestion event occurs at the state $2$:$(M,V_2,0)$ through the path $2$:$(M,V_2,0) \goes{}_B 3$:$(M,V_2,0) \goes{}_B 8$:$(M,V_1,1)$. States $2$ and $3$ are not distinguishable by the values of their observables, but intuitively $3$ models a state where a congestion event has been somehow notified to the controller. In a labelled structure (e.g.\ LTS) the path above could have expressed by a single labelled transition like $2$:$(M,V_2,0) \goes{cong} 8$:$(M,V_1,1)$, which would be enabled, for instance, when a ``congestion signal'' is received. Encoding a labelled structure into an unlabelled one requires indeed such an intermediate state, as shown in~\cite{de1990action} in the case of mapping LTSs into Kripke structures. Similarly, the event of traffic returning to normal conditions is simulated at state $10$:$(S,V_0,1)$ through the path $10$:$(S,V_0,1) \goes{}_B 13$:$(S,V_0,1) \goes{}_B 4$:$(S,V_0,0)$. In principle, such events could occur anywhere in the system, but, for the sake of simplicity, they have been implemented just at the above specified states.

\begin{figure}
\centering
\includegraphics[width=\textwidth]{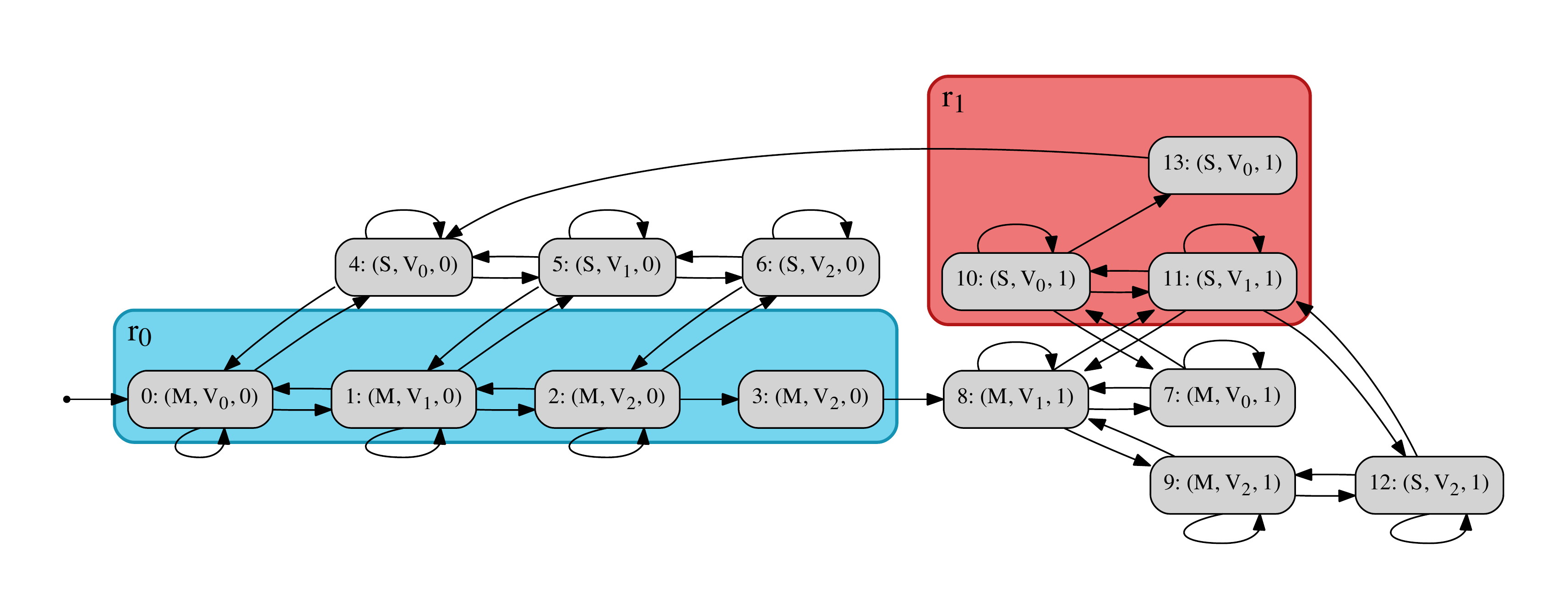}
\caption{The behavioural state machine for the motion control example. Each state is $s$ labelled by different evaluation of the observable variables $s$:$(r_s,v_s,c_s)$, i.e.\ state:(road, velocity, congestion). Coloured areas are used to represent the states of the $S$ level, which identify stable regions in the $B$ level. The $S$ states considered are: $r_0$ (normal mode, blue) and $r_1$ (fallback/congestion mode, red).}
\label{fig:b-level_atv}
\end{figure}

\subsubsection*{Structural Level}
We consider two different structural levels, $S_0$ (Fig.~\ref{fig:s-levels_atv} (a)) and $S_1$ (Fig.~\ref{fig:s-levels_atv} (b)). Ideally, these $S$ levels regulates the different modes of operation of the ATV and consists of the following $S$ states:
\begin{itemize}
\item[$r_0:$] it corresponds to the normal mode of operation, occurring when the main road is driven ($r==M$) and there is no traffic congestion ($c==0$); and
\item[$r_1:$] it models the fallback mode, occurring when congestion occurs ($c==1$); in this case the ATV has to be in the secondary road ($r==S$), which implies that it cannot drive at the maximum velocity ($v==V_0 \vee v==V_1$).
\end{itemize}
Figure~\ref{fig:b-level_atv} shows the sub-behaviours of the $B$ level, as identified by the $S$ states $r_0$ and $r_1$.

$S_0$ describes the adaptation between the normal mode and the fallback mode in case of traffic congestion, and back from the fallback mode to the normal one, when the motion controller is notified that congestion is over.
The structural state machine $S_0$ is given by:
$$S_0 = (\{r_0, r_1\}, r_0, \mathcal{O}^{\Sigma,A}_{M}, \{ r_0 \goes{v==V_0 \vee v==V_1}_S r_1, r_1 \goes{c==0}_S r_0 \}, L ),$$
where $\mathcal{O}^{\Sigma,A}_{M}$ is the above introduced observation function; and $L$ is the labelling function giving the previously explained constraints. Below we discuss in more detail the transitions of $S_0$.
\begin{itemize}
\item $r_0 \goes{v==V_0 \vee v==V_1}_S r_1.$ According to the transition invariant, during the adaptation between the normal and the fallback mode, the ATV must not drive at the maximum speed.
\item $r_1 \goes{c==0}_S r_0.$ In order to adapt from the fallback back to the normal mode, it must always hold that congestion is over.
\end{itemize}

On the other hand, $S_1$ does not constraint the behaviour of the ATV by including the fallback $S$ state $r_1$ when a congestion event happens. The corresponding state machine is
$$S_1 = (\{r_0\}, r_0, \{ r_0 \goes{v==V_0 \vee v==V_1}_S r_0 \}, L ),$$
where the transition $r_0 \goes{v==V_0 \vee v==V_1}_S r_0$ indicates that while adapting back to the normal mode $r_0$, the system has to meet the velocity constraints specified in the adaptation invariant ($v==V_0 \vee v==V_1$).

\begin{figure}
\centering
\subfloat[]{\includegraphics[width=0.65\textwidth]{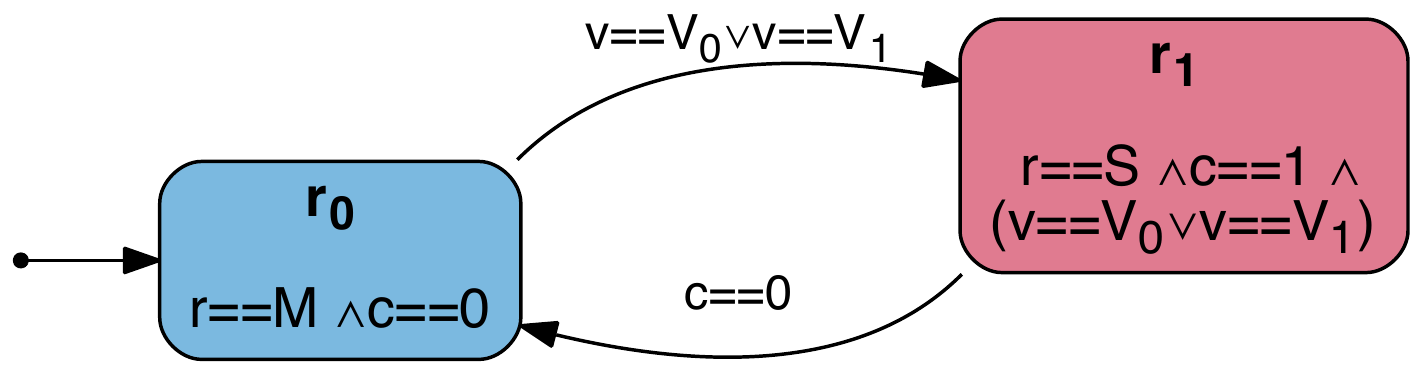}}\hfill
\subfloat[]{\includegraphics[width=0.24\textwidth]{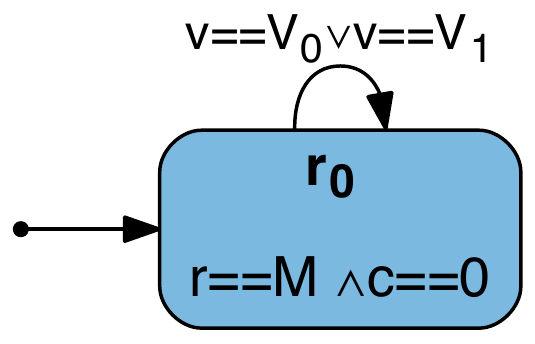}}
\caption{The two different structural levels for the motion controller example. $S_0$ (Fig.~\ref{fig:s-levels_atv} (a)) models the adaptation logic between two operation modes, $r_0$ (normal) and $r_1$ (fallback). Instead $S_1$ just consider adaptations from $r_0$ to itself, without including intermediate $S$ states.}
\label{fig:s-levels_atv}
\end{figure}
\section{Operational Semantics of the Flat $S[B]$ system}
\label{sect:semantics}
In this section we give the operational semantics of an $S[B]$ system as a transition system resulting from the flattening of the behavioural and the structural levels. We obtain a Labelled Transition System (LTS) over states of the form $(q, r, \rho)$, where
\begin{itemize}
\item $q \in Q$ and $r \in R$ are the active $B$ state and $S$ state, respectively; and
\item $\rho$ keeps the target $S$ state that must be reached during adaptation and the invariant that must be fulfilled during this phase. Therefore $\rho$ is either empty (no adaptation is occurring), or a singleton $\lbrace (\psi, r') \rbrace$, with $\psi \in \Psi(\Sigma,A)$ a formula and $r' \in R$ an $S$ state.
\end{itemize}
\begin{defn}[Flat $S\text{[}B\text{]} system$]
Consider an $S[B]$ system. The corresponding flat $S[B]$ system is the LTS
$\mathcal{F}(S[B])=(F,f_0,\goes{r} \cup \goes{r,\psi,r'})$ where
\begin{itemize}
\item $F = Q \times R \times (\{(\psi,r') \mid \exists r \in R . \ r \goes{\psi}_S r' \} \cup \{\emptyset \})$
is the set of states; 
\item $f_0 = (q_0, r_0, \emptyset)$ is the initial state;
\item $\goes{r} \subseteq F \times F$, with $r \in R$, is a family of transition relations between non-adapting states, i.e., both satisfying $L(r)$; 
\item $\goes{r,\psi,r'} \subseteq F \times F$, with $r,r' \in R$ and $\psi \in \Psi(\Sigma,A)$, is a family of transition relations between adapting states, where the adaptation is determined by the $S$ transition $r \goes{\psi}_S r'$; and 
\item the pairs in $\goes{r}$ and in $\goes{r,\psi,r'}$ are all and only those derivable using the rules in Table~\ref{tbl:sos}.
\end{itemize}
\end{defn}

\begin{table}
\begin{small}
\[
\begin{array}{|c|}
\hline
\name{Steady} \quad \sos{q \models L(r) \quad q \goes{}_B q' \quad q' \models L(r)}{(q, r, \emptyset) \goes{r} (q', r, \emptyset)}
\\
\name{AdaptStart} \quad 
\sos
{\forall q''.(q \goes{}_Bq'' \implies q'' \not\models L(r)) \\  
q \models L(r) \quad q \goes{}_B q' \quad r \goes{\psi}_S r'  \quad  q' \not \models L(r') \quad q' \models \psi}
{( q, r, \emptyset) \goes{r,\psi,r'} ( q', r, \lbrace(\psi,r')\rbrace)}
\\
\name{Adapt} \quad 
\sos
{\forall q''.(q \goes{}_Bq'' \implies q'' \not\models L(r'))\\  
q \models \psi \quad  q \not\models L(r') \quad q \goes{}_B q'  \quad q' \models \psi }
{( q, r, \lbrace(\psi,r')\rbrace) \goes{r,\psi,r'} ( q', r, \lbrace(\psi,r')\rbrace)}
\\
\name{AdaptEnd} \quad 
\sos
{q \models \psi \quad  q \not\models L(r') \quad q \goes{}_B q'  \quad q' \models L(r')}
{(q, r, \lbrace(\psi,r')\rbrace) \goes{r,\psi,r'} ( q', r', \emptyset)}
\\
\name{AdaptStartEnd} \quad 
\sos
{\forall q''.(q \goes{}_Bq'' \implies q'' \not\models L(r)) \\  
q \models L(r) \quad q \goes{}_B q' \quad r \goes{\psi}_S r'  \quad  q' \models L(r') \quad}
{( q, r, \emptyset) \goes{r,\psi,r'} ( q', r', \emptyset)}
\\
\hline
\end{array}\]
\end{small}
\caption{Operational semantics of the flat $S[B]$ system}\label{tbl:sos}
\end{table}

Let us discuss the rules listed in Table~\ref{tbl:sos} characterising the flattened transitional semantics:
\begin{itemize}
\item Rule \textsc{Steady} describes the steady (i.e.\ non-adapting) behaviour
of the system. If the system is not adapting and a $B$ state $q$ can perform a
transition to a $q'$ that satisfies the current constraints $L(r)$, then the
flat system can perform a non-adapting transition $\goes{r}$ of the form $(q, r,
\emptyset) \goes{r} (q', r, \emptyset)$.
\item Rule \textsc{AdaptStart} regulates the starting of an adaptation phase.\\Adaptation occurs when all  of the next $B$ states do not satisfy the current $S$ state constraints - i.e.\ $\forall q''.(q \goes{}_Bq'' \implies q'' \not \models L(r)$ - and the $B$ machine is not itself deadlocked ($q \goes{}_{B} q'$). In this case, for each $S$ transition $r \goes{\psi}_S r'$ an adaptation towards the target state $r'$, under the invariant $\psi$, can start. The flat system performs an adapting transition $\goes{r,\psi,r'}$ of the form $(q, r, \emptyset) \goes{r,\psi,r'} (q', r, \lbrace(\psi,r')\rbrace)$.
\item Rule \textsc{Adapt} can be used only during an adaptation phase. It handles the case in which, after the current transition, the system keeps adapting because a steady configuration cannot be reached ($\forall q''.(q \goes{}_Bq'' \implies q'' \not\models L(r'))$). In this situation, since the system still must adapt ($q \not \models L(r')$), if the $B$ machine is not deadlocked and the invariant can still be satisfied ($q\goes{}_Bq'$ and $q' \models \psi $), the rule allows a transition of the form $(q, r, \lbrace(\psi,r')\rbrace)
\goes{r,\psi,r'} (q', r, \lbrace(\psi,r')\rbrace)$. Note that during adaptation the
behaviour is not regulated by the $S$ states constraints. Note also that
the semantics does not assure that a state where the target $S$ state constraints 
hold is eventually reached. Two different formulations of such adaptability requirements are
given in Section~\ref{sect:adaptability}.
\item Also rule \textsc{AdaptEnd} can only be applied during an adaptation phase and it handles the case in which, after the current transition, the adaptation must end because a steady configuration has been reached ($q' \models L(r')$). It allows a transition $\goes{r,\psi,r'}$ from an adapting state $(q, r, \lbrace(\psi,r')\rbrace)$ to the steady (non-adapting) state $(q', r', \emptyset)$.
\item Rule \textsc{AdaptStartEnd} handles the special case in which an adaptation phase must start from a steady situation - $\forall q''.(q \goes{}_Bq'' \implies q'' \not \models L(r)$ - but then, after just one move of the $B$ level, another steady region of the $S$ level is reached ($ q' \models L(r')$). In this case the invariant $\psi$ associated to the $S$ transition is ignored and the system goes directly into another steady state. Note that this rule is alternative to the rule \textsc{AdaptStart} in which the initial situation is the same, but the steady region is not reached after one $B$ transition. The flat system performs an adapting transition $\goes{r,\psi,r'}$ of the form $(q, r, \emptyset) \goes{r,\psi,r'} (q', r,' \emptyset)$.
\end{itemize}

Let us now state some properties of the given flat semantics. In the following, given any transition relation $\rightarrow$ and any state $s$, by $s \rightarrow$ and by $s \not \rightarrow$ we mean, as usual, that there exists a state $s'$ such that $s \rightarrow s'$ and that there exists no state $s'$ such that $s \rightarrow s'$, respectively. Moreover, by  $\rightarrow^{+}$ we indicate a finite, non-empty, sequence of $\rightarrow$ steps; more formally, there exists $n \in \mathbb{N}, n > 0$ such that $s= s_{0} \rightarrow s_{1} \rightarrow \cdots s_{n-1} \rightarrow s_{n}$. Finally, by $\rightarrow^{k}$, $k\geq 0$, we indicate $k$ consecutive steps of the relation $\rightarrow$:  $s= s_{0} \rightarrow s_{1} \rightarrow \cdots s_{k-1} \rightarrow s_{k}$. If $k = 0$, then $s \rightarrow^{0} s'$ is equivalent to say that there is the empty sequence of steps $s$, and thus $s'=s$. This is always possible, even if the relation $\rightarrow$ is not reflexive. 

\begin{prop}[Properties of flat semantics]\label{prop:semantics} \ \\
Let $\mathcal{F}(S[B])=(F,f_0,\goes{r} \cup \goes{r,\psi,r'})$ be a flat $S[B]$ system. All the following statements hold:
\begin{itemize}
\item[(i)] If a steady transition can be performed, then adaptation cannot start:\\
$\forall (q,r, \emptyset) \in F. \; (q,r,\emptyset) \goes{r} (q',r,\emptyset) \implies (q,r,\emptyset) \goes{r, \psi, r'} \! \! \! \! \! \! \not $

\item[(ii)] If adaptation can start, then no steady transition is possible:\\
$\forall (q,r, \emptyset) \in F. \; (q,r,\emptyset) \goes{r, \psi, r'} (q',r,\{(\psi,r')\}) \implies (q,r,\emptyset) \not \goes{r} $

\item[(iii)] During adaptation no steady transition is possible:\\
$\forall (q',r',\{(\psi,r')\}) \in F. \; (q',r',\{(\psi,r')\})  \not \goes{r}$

\item[(iv)] The non-adapting and the adapting transition relations are disjoint:\\$\forall r,r' \in R, \forall \psi \in \Psi(\Sigma,A). \; \goes{r} \cap \goes{r,\psi,r'} = \emptyset$


\item[(v)] In case of a successful adaptation, the adaptation phase ends as soon as possible, i.e.\ as soon as the target steady state can be reached with a single transition.

\item[(vi)] Given any  $q \in Q$ and $r \in R$, then every path $\pi$ in $\mathcal{F}(S[B])$ starting in a state $(q,r,\emptyset)$ is of one of the following two kinds:

\begin{description}
\item[(1)] $\pi$ is finite (possibly empty) or $\pi$ is infinite and it has the form:
$$
\begin{array}{l}
\pi = (q= q_0, r=r_0,\emptyset)
(\goes{r_{0}})^{m_0}(\goes{r_{0},\psi_{0},r_{1}})^{n_0}\cdots \\
\cdots (q_i,r_i,\emptyset)
 (\goes{r_{i}})^{m_i}(\goes{r_{i},\psi_{i},r_{i+1}})^{n_i}
 (q_{i+1},r_{i+1},\emptyset) \cdots \\
\cdots (q_{k-1},r_{k-1},\emptyset) (\goes{r_{k-1}})^{m_{k-1}}(\goes{r_{k-1},\psi_{k-1},r_{k}})^{n_{k-1}}(q_{k}, r_{k}, \emptyset) \cdots
\end{array}
$$
where $k \geq 0$ and for each $i \geq 0$, either $m_i = 1 \wedge n_i = 0$ (steady transition) or $m_i = 0 \wedge n_i > 0$ (adaptation path);
\item[(2)] $\pi$ is finite, non-empty and it has the form:
$$
\begin{array}{l}
\pi = (q= q_0, r=r_0,\emptyset)
(\goes{r_{0}})^{m_0}(\goes{r_{0},\psi_{0},r_{1}})^{n_0}\cdots \\
\cdots (q_i,r_i,\emptyset)
 (\goes{r_{i}})^{m_i}(\goes{r_{i},\psi_{i},r_{i+1}})^{n_i}
 (q_{i+1},r_{i+1},\emptyset) \cdots \\
\cdots (q_k,r_k,\emptyset) (\goes{r_{k},\psi,r'})^{n_k} (q',r_{k}, \{(\psi, r')\}) 
\end{array}
$$
where $k \geq 0$, for each $0\leq i < k$ either $m_i = 1 \wedge n_i = 0$  or $m_i = 0 \wedge n_i > 0$, $n_{k} > 0$ and $(q',r_{k}, \{(\psi, r')\}) \;\;\;\;  \not \!\!\!\!\!\!\!\!\!\! \goes{r_{k}, \psi, r'}$. In this case the path stops during adaptation.
\end{description}

\item[(vii)] Let $\pi \in \mathcal{F}(S[B])$ be a path starting in a state $(q,r,\emptyset)$ such that $q \models L(r)$. Then, in every position $i$ of the path such that $\pi[i] = (q_{i},r_{i},\emptyset)$, it holds $q_{i} \models L(r_{i})$.
\end{itemize}
\end{prop}
\begin{proof}
See~\ref{proof:prop1}
\end{proof}

\subsection{Termination}\label{subsect:term}
In an $S[B]$ system, termination cannot be compatible with adaptability. We see adaptability as the property for which a system \textit{continuously} operates under stable, allowed modes (steady states), by possibly perform adaptation paths across modes. 

In the flat semantics deadlocks occurring at adapting states, e.g.\ when the adaptation invariant cannot be met, are clearly conflicting with the concept of adaptability. Instead, deadlocks at steady states are more subtle to interpret, since they may occur under two different conditions:
\begin{itemize}
\item the current constraints cannot be satisfied by the next states of the current state, but, at the same time, adaptation cannot start because none of the next $B$ states meet any of the adaptation invariants and any of the target constraints. In other words, the flat semantics terminates even if the $B$ level can proceed. Evidently, this violates adaptability.
\item the $B$ level cannot progress at all. We consider this situation as a \textit{bad deadlock} state in the behavioural model. Conversely, every $B$ state indicating a \textit{good termination} should have the chance to progress and therefore must be modelled, as usual in this case, with an idling self-loop.
\end{itemize}

We capture the requirement for which the flat $S[B]$ must not terminate through the $\textsc{Progress}(q,r)$ predicate:
$$\textsc{Progress}(q,r) \iff (q,r,\emptyset)\goes{r} \ \vee \ (q,r,\emptyset)\goes{r,\psi,r'}$$

\subsection{Flat Semantics of the Motion Control Example}
The flat semantics of the two systems $S_0[B]$ and $S_1[B]$ implementing the ATV motion controller case study is depicted in Figure~\ref{fig:sbs_atv}.

Notably, the same behavioural level $B$ possesses different adaptation capabilities depending on the structure $S$ that is considered. Indeed, in $\mathcal{F}(S_0[B])$ every adaptation path leads to a target $S$ state. On the other hand, in $\mathcal{F}(S_1[B])$ there always exists an adaptation path leading to a target stable region, but it contains cycles of adapting states, thus leading to infinite adaptation paths. 

In other words, the behavioural level $B$ is able to successfully adapt under the structural level $S_0$, for \textit{all possible adaptation paths}. Thus, recalling the definitions introduced in Sect.~\ref{sect:intro}, $S_0[B]$ is \textit{strong adaptable}. Conversely, $B$ is able to successfully adapt under $S_1$, only for \textit{some adaptation paths}, i.e.\ the finite ones. Therefore, $S_1[B]$ is \textit{weak adaptable}. These two different kinds of adaptability are formalized in Section~\ref{sect:adaptability}.

\begin{figure}
\centering
\subfloat[]{\includegraphics[width=0.16\textwidth]{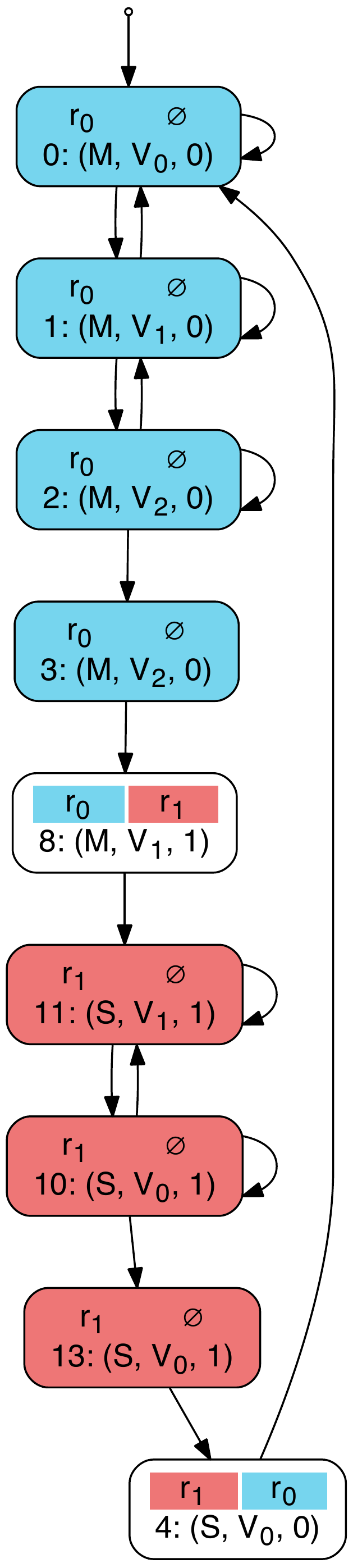}}\hfill
\subfloat[]{\includegraphics[width=0.3\textwidth]{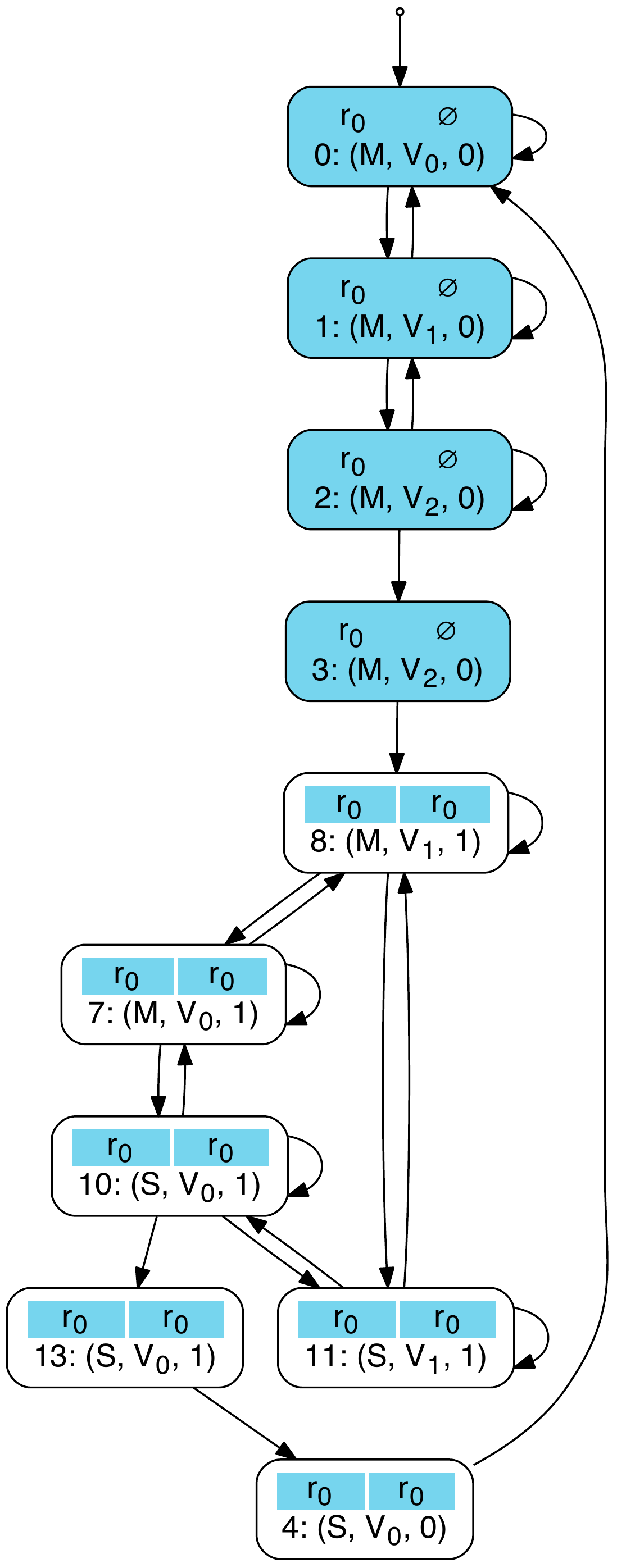}}
\caption{The flat semantics of the two systems $S_0[B]$ (Fig.~\ref{fig:sbs_atv} (a)) and $S_1[B]$ (Fig.~\ref{fig:sbs_atv} (b)) in the motion control model. For the sake of clarity, transition labels have been omitted. Since in both $S_0$ and $S_1$ there is at most one transition between two $S$ states, transition invariants have not been displayed, but can be found in the description of the model. Filled states represent steady (i.e.\ non-adapting) states, while void states represent adapting states (the current $S$ state and the target $S$ state are displayed in the top-left and top-right, respectively).
}
\label{fig:sbs_atv}
\end{figure}

\section{Adaptability Properties}
\label{sect:adaptability}
The transitional semantics introduced in Section~\ref{sect:semantics} does not guarantee that an adaptation phase can always start or that, once started, it always ends up in a state satisfying the constraints of the target $S$ state. In this section we want to give some formal tools to analyse a given system w.r.t.\ these kind of properties. As a first step, we characterise two adaptability notions by means of two relations over the set of $B$ states and the set of $S$ states, namely a \textit{weak adaptation relation} $\mathcal{R}_w$ and a \textit{strong adaptation relation} $\mathcal{R}_s$. Then, we characterise the same adaptability notions logically and we prove that they can be model checked by using proper formulae of a temporal logic. 


Informally, a state $q$ of $B$ is \emph{weak adaptable} to a state $r$ of $S$ if it satisfies the constraints imposed by $r$ and some of its successors are either weak adaptable to the same $r$ or there is an adaptation phase of the flat system that, from $q$, reaches a state $q'$ that is weak adaptable to another $S$ state $r'$. In other words, we require that states satisfying $L(r)$ are in relation with $r$ and that from ``border'' states, that is to say those that can start an adaptation phase for leaving $r$, there is always at least one way to safely reach another steady situation in another $S$ state $r'$. As explained in Sect.~\ref{subsect:term}, we prevent adaptation from terminating by applying the $\textsc{Progress}$ predicate that avoids those steady configurations in a state of bad termination. We formally define this relation using a co-inductive style as it is usually done, for instance, for bisimulation relations.

\begin{defn}[Weak adaptation]\label{def:weakadapt} Given an $S[B]$ system, a binary relation $\mathcal{R} \subseteq Q \times R$ is a \emph{weak adaptation} if and only if whenever $q \; \mathcal{R} \; r$ we have:
\begin{itemize}
\item[(i)] $q \models L(r)$ and $\textsc{Progress}(q,r)$, \emph{and}
\item[(ii)] if $(q,r,\emptyset) \goes{r}$ then there exists $q' \in Q$  such that $(q,r,\emptyset) \goes{r} (q',r,\emptyset)$ and $q' \; \mathcal{R} \; r$, \emph{and}
\item[(iii)] if $(q,r,\emptyset) \goes{r,\psi',r''}$ for some $\psi' \in \Psi(\Sigma,A)$ and $r'' \in R$ then there exist $q' \in Q, \ \psi \in \Psi(\Sigma,A)$ and $r' \in R$ such that $(q, r, \emptyset)\goes{r, \psi, r'}^{+}  \! (q', r', \emptyset)$ and $q' \; \mathcal{R} \; r'$.
\end{itemize}
We say that a state $q \in Q$ is \emph{weak adaptable} to a state $r \in R$, written $q|_{w} r$, if and only if there is a weak adaptation relation $\mathcal{R}$ such that $(q,r)$ is in $\mathcal{R}$.

At the level of the whole system, we say that $S[B]$ is \emph{weak adaptable} if the
initial $B$ state $q_{0}$ is weak adaptable to the initial $S$ state $r_{0}$.
\end{defn}

\begin{prop}[Union of Weak Adaptation Relations] Given an $S[B]$ system, if $\mathcal{R}_{1}$ and $\mathcal{R}_{2}$ are weak adaptation relations, then $\mathcal{R}_{1} \cup \mathcal{R}_{2}$ is a weak adaptation relation.
\end{prop}
\begin{proof}
See~\ref{proof:prop2}
\end{proof}

\begin{defn}[Weak adaptability] Given an $S[B]$ system, the union of all weak adaptation relations among the states $Q$ and $R$ of $S[B]$ is denoted by $\mathcal{R}_{w}$ and is the \emph{weak adaptability} relation of $S[B]$.
\end{defn}

\begin{lem}[Propagation of Weak Adaptation Relation]\label{lem:propagation_weak} Consider an $S[B]$ system and let $q$ and $r$ be such that $q |_{w} r$. Then there exists in $\mathcal{F}(S[B])$ an infinite path
$$
\begin{array}{l}
\pi = (q= q_0, r=r_0,\emptyset)
(\goes{r_{0}})^{m_0}(\goes{r_{0},\psi_{0},r_{1}})^{n_0} \cdots \\
\cdots (q_i,r_i,\emptyset)
 (\goes{r_{i}})^{m_i}(\goes{r_{i},\psi_{i},r_{i+1}})^{n_i}
 (q_{i+1},r_{i+1},\emptyset) \cdots 
 \end{array}
$$
such that $\forall i \geq 0 \; . \; q_{i} |_{w} r_{i} \wedge ((m_i = 1 \wedge n_i = 0) \vee (m_i = 0 \wedge n_i > 0))$.
\end{lem}
\begin{proof}
See~\ref{proof:lemma1}
\end{proof}

Weak adaptability guarantees that there is always at least one way for a certain state of an $S[B]$ system to adapt, that is to say to continue to evolve in a consistent way w.r.t.\ the structural constraints of the $S$ level. A stronger property that could be useful to know about the adaptability of a system is what we call \emph{strong adaptability}. A $B$ state $q$ is \emph{strong adaptable} to an $S$ state $r$ if it satisfies the constraints imposed by $r$ and all its successors $q'$ are either strong adaptable to the same $r$ or they are always the starting point of a successful adaptation phase towards states $q''$ that are strong adaptable to other $S$ states. Again, bad deadlocks are excluded from the relations. This time we require that all the ``border'' states are safe doors to other steady situations, whatever path is taken from them.

\begin{defn}[Strong adaptation]\label{def:strongadapt}
Given an $S[B]$ system, a binary relation $\mathcal{R} \subseteq Q \times R$ is a \emph{strong adaptation} if and only if whenever $q \; \mathcal{R} \; r$ we have:
\begin{itemize}
\item[(i)] $q \models L(r)$ and $\textsc{Progress}(q,r)$, \emph{and}
\item[(ii)] for all $q' \in Q$,  if $(q,r,\emptyset) \goes{r} (q',r,\emptyset)$ then $q' \; \mathcal{R} \; r$, \emph{and}
\item[(iii)] all paths of the form
$$
(q,r,\emptyset) \goes{r,\psi,r'} (q_{1}, r, \{(\psi, r') \} )\goes{r,\psi,r'} 
\cdots 
(q_{i}, r, \{(\psi, r') \}) \goes{r,\psi,r'} \cdots
$$ 
are finite and end up in a state $(q', r', \emptyset)$ such that $q' \; \mathcal{R} \; r'$.
\end{itemize}
We say that a state $q \in Q$ is \emph{strong adaptable} to a state $r \in R$, written $q|_{s} r$, if and only if there is a strong adaptation relation $\mathcal{R}$ such that $(q,r)$ is in $\mathcal{R}$.

At the level of the whole system, we say that $B$ is \emph{strong adaptable} to $S$ if the
initial $B$ state $q_{0}$ is strong adaptable to the initial $S$ state $r_{0}$.
\end{defn}

\begin{prop}[Union of Strong Adaptation Relations]\label{prop:unionstrong} Given an $S[B]$ system, if $\mathcal{R}_{1}$ and $\mathcal{R}_{2}$ are strong adaptation relations, then $\mathcal{R}_{1} \cup \mathcal{R}_{2}$ is a strong adaptation relation.
\end{prop}
\begin{proof}
As in the case of weak adaptation relations.
\end{proof}

\begin{defn}[Strong adaptability] Given an $S[B]$ system, the union of all strong adaptation relations among the states $Q$ and $R$ of $S[B]$ is denoted by $\mathcal{R}_{s}$ and is the \emph{strong adaptability} relation of $S[B]$.
\end{defn}

In the remainder of the paper we will alternatively say that an $S[B]$ system is weak (strong) adaptable, in the sense that $B$ is weak (strong) adaptable to $S$. It is straightforward to see that strong adaptability implies weak adaptability, since the strong version of the relation requires that every adaptation path reaches a target $S$ state, while the weak version just requires that at least one adaptation path reaches a target $S$ state. 

\begin{prop}[Strong Adaptation implies Weak Adaptation] Consider an $S[B]$ system and let $q$ and $r$ be such that $q |_{s} r$. Then, it holds $q |_{w} r$.
\end{prop}
\begin{proof}
See~\ref{proof:prop4}
\end{proof}

Given the flat semantics $\mathcal{F}(S[B])=(F,f_0,\goes{r} \cup \goes{r,\psi,r'})$ of an $S[B]$ system, we will denote, in the following, the set of reachable states from a certain state $f \in F$ as the reflexive and transitive closure $Post^*(f)$ of the operator $Post(s) = \{ s' \in F \ | \ (s,s') \in \ \goes{r} \cup \goes{r,\psi,r'}\}$.

\begin{lem}[Propagation of Strong Adaptation Relation]\label{lem:propagation} Consider an $S[B]$ system and let $q$ and $r$ be such that $q |_{s} r$. Then, every state $(q',r',\emptyset) \in Post^{*}((q,r,\emptyset))$ is such that $q' |_{s} r'$.
\end{lem}
\begin{proof}
See~\ref{proof:lemma2}
\end{proof}

The following proposition gives a precise candidate relation for checking if a system is strong adaptable: such a candidate is determined by the steady states of the flat semantics that are reachable from the initial state.

\begin{prop}[Construction of Strong Adaptation Relation]\label{prop:prop_adapt}
Given an $S[B]$ system, let $\mathcal{F}(S[B])=(F,f_0,\goes{r} \cup \goes{r,\psi,r'})$ be its flat semantics. 
Then $S[B]$ is strong adaptable if and only if $\mathcal{R} = \{ (q,r) \in Q \times R \mid (q,r,\emptyset) \in Post^{*}(f_{0})\}$ is a strong adaptation relation.
\end{prop}
\begin{proof}
See~\ref{proof:prop5}
\end{proof}

%
\subsection{Adaptation Relations in the Motion Control Example}
In the following we will show that in the ATV motion control case study $S_0[B]$ is strong adaptable (and thus also weak adaptable) and $S_1[B]$ is weak adaptable, but not strong adaptable.

In order to verify that $S_0[B]$ is strong adaptable, we need to prove that $q_0|_{s} r_0$, by finding a strong adaptation relation $\mathcal{R}$ s.t.\ $(q_0, r_0) \in \mathcal{R}$. Note that in $\mathcal{F}(S_0[B])$, every state is reachable from the initial state $(0, r_0, \emptyset)$. Therefore by Proposition~\ref{prop:prop_adapt}, we consider the relation $\mathcal{R} =\{(q,r) \ | \ (q,r,\emptyset) \in F\}$, where $F$ is the set of flat states of $\mathcal{F}(S_0[B])$:
\[
\mathcal{R} = \{ (0, r_0), (1, r_0), (2, r_0), (3, r_0), (11, r_1), (10, r_1), (13, r_1)\}.
\]
It is easy to verify that $\forall (q,r) \in \mathcal{R}. \ q \models L(r)$; and that $\textsc{Progress}(q,r)$ holds for any of such states, because there are no deadlock states in the flat semantics. Therefore condition $(i)$ of the definition of strong adaptation is always true and has not to be further checked. Clearly, $(q_0, r_0) \in \mathcal{R}$. We show that $\mathcal{R}$ is a strong adaptation relation, by checking requirements $(ii)$ and $(iii)$ of Def.~\ref{def:strongadapt} for each element of $\mathcal{R}$.
\begin{small}
\begin{itemize}
\item $(0$:$(M,V_0,0), r_0)$.
\begin{itemize}
\item[$(ii)$] $(0, r_0, \emptyset)\goes{r_0}(0, r_0, \emptyset)$ and $(0, r_0) \in \mathcal{R}$; $(0, r_0, \emptyset)\goes{r_0}(1, r_0, \emptyset)$ and \\$(1, r_0) \in \mathcal{R}$
\item[$(iii)$] $(0, r_0, \emptyset)\goes{r,\psi,r'}\! \! \! \! \! \! \not$
\end{itemize}
\item $(1$:$(M,V_1,0), r_0)$.
\begin{itemize}
\item[$(ii)$] $(1, r_0, \emptyset)\goes{r_0}(0, r_0, \emptyset)$ and $(0, r_0) \in \mathcal{R}$; $(1, r_0, \emptyset)\goes{r_0}(1, r_0, \emptyset)$ and \\$(1, r_0) \in \mathcal{R}$; $(1, r_0, \emptyset)\goes{r_0}(2, r_0, \emptyset)$ and $(2, r_0) \in \mathcal{R}$
\item[$(iii)$] $(1, r_0, \emptyset)\goes{r,\psi,r'}\! \! \! \! \! \! \not$
\end{itemize}
\item $(2$:$(M,V_2,0), r_0)$.
\begin{itemize}
\item[$(ii)$] $(2, r_0, \emptyset)\goes{r_0}(1, r_0, \emptyset)$ and $(1, r_0) \in \mathcal{R}$; $(2, r_0, \emptyset)\goes{r_0}(2, r_0, \emptyset)$ and\\$(2, r_0) \in \mathcal{R}$; $(2, r_0, \emptyset)\goes{r_0}(3, r_0, \emptyset)$ and $(3, r_0) \in \mathcal{R}$
\item[$(iii)$] $(2, r_0, \emptyset)\goes{r,\psi,r'}\! \! \! \! \! \! \not$
\end{itemize}
\item $(3$:$(M,V_2,0), r_0)$.
\begin{itemize}
\item[$(ii)$] $(3, r_0, \emptyset)\not\goes{r}$
\item[$(iii)$] there is only one adaptation path from $(3, r_0, \emptyset)$ leading to the flat state $(11, r_1, \emptyset)$, and $(11, r_1) \in \mathcal{R}$.
\end{itemize}
\item $(11$:$(S, V_1, 1), r_1)$.
\begin{itemize}
\item[$(ii)$] $(11, r_1, \emptyset)\goes{r_1}(11, r_1, \emptyset)$ and $(11, r_1) \in \mathcal{R}$; $(11, r_1, \emptyset)\goes{r_1}(10, r_1, \emptyset)$ and $(10, r_1) \in \mathcal{R}$
\item[$(iii)$] $(11, r_1, \emptyset)\goes{r,\psi,r'}\! \! \! \! \! \! \not$
\end{itemize}
\item $(10$:$(S, V_0, 1), r_1)$.
\begin{itemize}
\item[$(ii)$] $(10, r_1, \emptyset)\goes{r_1}(11, r_1, \emptyset)$ and $(11, r_1) \in \mathcal{R}$; $(10, r_1, \emptyset)\goes{r_1}(10, r_1, \emptyset)$ and $(10, r_1) \in \mathcal{R}$; $(10, r_1, \emptyset)\goes{r_1}(13, r_1, \emptyset)$ and $(13, r_1) \in \mathcal{R}$
\item[$(iii)$] $(10, r_1, \emptyset)\goes{r,\psi,r'}\! \! \! \! \! \! \not$
\end{itemize}
\item $(13$:$(S,V_0,1), r_1)$.
\begin{itemize}
\item[$(ii)$] $(13, r_1, \emptyset)\not\goes{r}$
\item[$(iii)$] there is only one adaptation path from $(13, r_1, \emptyset)$ leading to the flat state $(0, r_0, \emptyset)$, and $(0, r_0) \in \mathcal{R}$.
\end{itemize}
\end{itemize}
\end{small}
On the other hand, we demonstrate that $S_1[B]$ is weak adaptable, by finding a weak adaptation relation $\mathcal{R}$ s.t.\ $(q_0, r_0) \in \mathcal{R}$. Consider the following relation:
\[
\begin{array}{rl}
\mathcal{R} = &\{(0, r_0), (1, r_0), (2, r_0), (3, r_0)\}.
\end{array}
\]
Similarly to $S_0[B]$, $(q_0,r_0) \in \mathcal{R}$ and for all $(q,r) \in \mathcal{R}$, $\ q \models L(r)$ and $\textsc{Progress}(q,r)$ both holds. Thus, we need to check requirements $(ii)$ and $(iii)$ of Def.~\ref{def:weakadapt} to prove that $\mathcal{R}$ is a weak adaptation relation. 

Actually, the elements $(0,r_0)$, $(1,r_0)$, $(2,r_0)$ of $\mathcal{R}$ meet the requirements $(ii)$ and $(iii)$ of the strong adaptation definition, as illustrated before. Thus, they also meet the weak requirements. The element $(3,r_0)$ complies with the weak adaptation definition, since $(3, r_0, \emptyset)\goes{r_0, \psi, r_0}^+(0, r_0, \emptyset)$ and $(0, r_0) \in \mathcal{R}$. However, $(3,r_0)$ cannot be in any strong relation (but by definition must be in the weak relation $\mathcal{R}$)  because there are infinite adaptation paths starting from it. It implies that $\mathcal{R}$ is a weak and not strong adaptation relation.

\section{Logical Characterisation of Adaptability Properties}
\label{sec:logic}
In this section we formulate the adaptability properties introduced in Section~\ref{sect:adaptability} in terms of formulae of a temporal logic that can be model checked \cite{katoen2008,clarke1999model}.

To this purpose we briefly recall the well-known \textit{Computation Tree Logic (CTL)}
\cite{clarke1986automatic,clarke81}, a branching-time logic whose semantics is defined in terms of paths along a Kripke structure \cite{kripke63}. Given a set $AP$ of atomic propositions, a Kripke structure is a tuple $(T,t_{0},\goes{}_{k},I)$ where $T$ is a finite set of states, $t_{0}$ is the initial state, $\goes{}_{\kappa} \subseteq T \times T$ is a left-total transition relation and $I \colon T \rightarrow 2^{AP}$ maps each state to the set of atomic propositions that are true in that state. Given a state $t \in T$, a path $\pi$ starting from $t$ has the form $\pi \colon t=t_{0} \goes{}_{\kappa} t_{1} \goes{}_{\kappa} t_{2} \goes{}_{\kappa} \cdots$, where for all $i = 1, 2, \ldots, (t_{i-1},t_{i}) \in \goes{}_{\kappa}$. Given a path $\pi$ and an index $i > 0$, by $\pi[i]$ we denote the $i$-th state along the path $\pi$. The set of all paths starting from $t$ is denoted by $\mathit{Paths}(t)$. Note that, since the transition relation is required to be left-total, all runs are infinite. To model a deadlocked or terminated state in a Kripke structure the modeller must put a self-cycle on that state.

The set of well-formed CTL formulae are given by the following grammar:
\[
\begin{array}{rl}
\phi ::= & \mathit{true} \ | \ p \ | \ \neg \phi  \ | \ \phi \wedge \phi   \ | \  \mathbf{AX} \phi \ | \ \mathbf{EX} \phi \ | \ \mathbf{A}[\phi \mathbf{U} \phi] \ | \ \mathbf{E}[\phi \mathbf{U} \phi]
\end{array}
\]
where $p \in AP$ is an atomic proposition, logical operators are minimal ($\neg, \wedge$) in order to generate all the usual ones, and temporal operators ($\mathbf{X}$ next, $\mathbf{U}$ until) quantify along paths and must be preceded by the universal path quantifier $\mathbf{A}$ or by the existential path quantifier $\mathbf{E}$. 

Given a state $t$ of the underlying Kripke structure, the satisfaction of a CTL formula $\phi$ in $t$, written $t \models_{\mathrm{CTL}} \phi$, is defined inductively as follows. 
$$
\begin{array}{lcl}
t \models_{\mathrm{CTL}} \mathit{true} & \mbox{    } & \mbox{for all } t\\
t \models_{\mathrm{CTL}} p & \mbox{  iff  } & p \in I(t)\\
t \models_{\mathrm{CTL}} \neg \phi & \mbox{  iff  } & t \not \models_{\mathrm{CTL}} \phi\\
t \models_{\mathrm{CTL}} \phi_{1} \wedge \phi_{2} & \mbox{  iff  } & t  \models_{\mathrm{CTL}} \phi_{1} \mbox{ and } t  \models_{\mathrm{CTL}} \phi_{2}\\
t \models_{\mathrm{CTL}} \mathbf{AX} \phi & \mbox{  iff  } & \forall \pi \in \mathit{Paths}(t) . \pi[1] \models_{\mathrm{CTL}} \phi \\
t \models_{\mathrm{CTL}} \mathbf{EX} \phi & \mbox{  iff  } & \exists \pi \in \mathit{Paths}(t) \colon \pi[1] \models_{\mathrm{CTL}} \phi \\
t \models_{\mathrm{CTL}} \mathbf{A}[\phi_{1} \mathbf{U} \phi_{2}] & \mbox{  iff  } & \forall \pi \in \mathit{Paths}(t) . \exists j \geq 0 \colon ( \pi[j] \models_{\mathrm{CTL}} \phi_{2} \mbox{ and }\\
 & & \;\;\;\;\;\;\;\;\;\;\;\;\;\;\;\;\;\;\;\;\;\;\;\;\;\;\;\;\;\;\;\;\;\;\; \forall  0 \leq i < j . \pi[i] \models_{\mathrm{CTL}} \phi_{1}) \\
t \models_{\mathrm{CTL}} \mathbf{E}[\phi_{1} \mathbf{U} \phi_{2}] & \mbox{  iff  } & \exists \pi \in \mathit{Paths}(t) \colon \exists j \geq 0 \colon ( \pi[j] \models_{\mathrm{CTL}} \phi_{2} \mbox{ and }\\
 & & \;\;\;\;\;\;\;\;\;\;\;\;\;\;\;\;\;\;\;\;\;\;\;\;\;\;\;\;\;\;\;\;\;\;\;\; \forall  0 \leq i < j . \pi[i] \models_{\mathrm{CTL}} \phi_{1}) \\
\end{array}
$$
Other useful temporal operators like $\mathbf{EF} \phi$ ($\phi$ holds potentially),  $\mathbf{AF} \phi$ ($\phi$ is inevitable), $\mathbf{EG} \phi$ (potentially always $\phi$) and $\mathbf{AG} \phi$ (invariantly $\phi$), are derived, as usual, as follows: $\mathbf{EF} \phi \equiv \mathbf{E}[\mathit{true} \mathbf{U} \phi]$, $\mathbf{AF} \phi \equiv \mathbf{A}[\mathit{true} \mathbf{U} \phi]$, $\mathbf{EG} \phi \equiv \neg \mathbf{AF} \neg \phi$ and $\mathbf{AG} \phi \equiv \neg \mathbf{EF} \neg \phi$.

In the following we provide a Kripke structure derived from the flat semantics $\mathcal{F}(S[B])$ and two CTL formulae characterising weak and strong adaptability. 

\begin{defn}[Associated Kripke structure]
Consider an $S[B]$ system and its associated flat semantics $\mathcal{F}(S[B]) =(F,f_0,\goes{r} \cup \goes{r,\psi,r'})$. Its \emph{associated Kripke structure} is defined as $\mathcal{K}(S[B]) =(T,t_{0},\goes{}_{k},I)$ where $T = F$, $t_{0} = f_{0}$, $\goes{}_{k} = \goes{r} \cup \goes{r,\psi,r'} \cup \; ( \goes{}_{\mathit{self}} \; \triangleq \{ (t,t) \mid t \not \! \goes{r} \wedge \; t \;\;\;\;\;  \not \!\!\!\!\!\!\!\!\!\! \goes{r_{s}, \psi, r'} \})$, $I$ is defined w.r.t.\ the set $AP = \{ adapting, steady, progress\}$ of atomic propositions as follows. For all $t \in T$:
\begin{description}
  \item[$(i)$] $adapting \in I(t) \iff t = (q,r,\rho) \wedge t \goes{r, \psi, r'}$ 
  \item[$(ii)$]  $steady \in I(t) \iff t = (q,r,\emptyset) \wedge (t \goes{r} \vee \; t \goes{r, \psi, r'})$
  \item[$(iii)$] $progress \in I(t) \iff  t \goes{r} \vee \; t \goes{r, \psi, r'}$
\end{description}
\end{defn}

Note that the only structural difference between $\mathcal{F}(S[B])$ and $\mathcal{K}(S[B])$ are the self-loop transitions in $\goes{}_{\mathit{self}}$ added in $\mathcal{K}(S[B])$. These are needed because the transition relation of the Kripke structure must be left-total, but indeed they allow us to keep the information that the states $t$ such that $t \goes{}_{\mathit{self}} t$ where originally deadlocked or bad terminated in $\mathcal{F}(S[B])$ (see discussion in Section~\ref{subsect:term}). Then, the atomic proposition $progress$, by its definition (identical to the one given for $\textsc{Progress}(q,r)$ at the end of Section~\ref{subsect:term}), is \emph{not} true in all and only those states of $\mathcal{K}(S[B])$ that were originally deadlocked or bad terminated in $\mathcal{F}(S[B])$. Moreover, we remark that in some states both $adapting$ and $steady$ propositions may hold at the same time. These are the already mentioned ``border'' states, i.e.\ those that are still in a steady situation, but will start adapting in the next transition. From these, the next state may be $steady$ again (immediate adaptation) or only $adapting$ (adaptation in more than one step).

The formulae that we will check on $\mathcal{K}(S[B])$ are the following:
\begin{itemize}
\item \textbf{Weak adaptation:} there is a path in which the progress condition continuously holds and, as soon as adaptation starts, there exists at least one path for which the system eventually ends the adaptation phase leading to a steady state.
\begin{equation}\label{eq:weak}
\mathbf{EG}((adapting \implies \mathbf{EF} \ steady) \; \wedge \; progress)
\end{equation}
\item \textbf{Strong adaptation:} for all paths, the progress condition always holds and  whenever the system is in an adapting state, from there all paths eventually ends the adaptation phase leading to a steady state.
\begin{equation}\label{eq:strong} 
\mathbf{AG}(( adapting \implies \mathbf{AF} \ steady) \; \wedge \; progress)
\end{equation}
\end{itemize}

We remark that the same formulae could be expressed in the Action-based Computation Tree Logic (ACTL) \cite{de1990action} without the need of defining the atomic propositions. However, we decided to use CTL because we recognise that it is one of the mostly known and used temporal logic for model checking, 

Moreover, we want to state that the expressive power given by CTL is adequate for the adaptability checking we introduce, in the sense that the same properties could not be expressed in the other mostly used logic, Linear Temporal Logic (LTL) \cite{pnueli77}. In particular, the weak adaptability property requires the existential path quantification ($\mathbf{EG}$, $\mathbf{EF}$), which cannot be expressed in LTL, to render the invariability of the possibility of adaptation along \emph{one} certain computation. Differently, the strong adaptability property could also be formulated in LTL as $\Box ((adapting \Rightarrow \Diamond steady) \wedge progress)$.

\begin{thm}[Weak adaptability checking]\label{teo:weak}\ \\
Consider an $S[B]$ system. Given a $B$ state $q$ and an $S$ state $r$ such that $ q \models L(r)$, then $q$ is weak adaptable to $r$ if and only if the weak adaptation CTL formula (equation~\ref{eq:weak}) is true in $\mathcal{K}(S[B])$ at state $(q,r,\emptyset)$. 
Formally, given a state $q \in [[L(r)]]$ 
$$
q \ |_w \ r 
\iff 
(q,r,\emptyset)\models_{\mathrm{CTL}} \mathbf{EG}((adapting \implies \mathbf{EF} \ steady) \; \wedge \; progress)
$$ 
\end{thm}
\begin{proof}
See~\ref{proof:thm1}
\end{proof}

\begin{cor}\ \\
Consider an $S[B]$ system. Then, $S[B]$ is weak adaptable if and only if 
$$
t_{0}  \models_{\mathrm{CTL}} \mathbf{EG}((adapting \implies \mathbf{EF} \ steady) \; \wedge \; progress)
$$ 
where $t_0$ is  the initial state of $\mathcal{K}(S[B])$.
 \end{cor}
 \begin{proof}
 The thesis follows easily from Definition~\ref{def:weakadapt} and from Theorem~\ref{teo:weak}.
 \end{proof}

\begin{thm}[Strong adaptability checking]\label{teo:strong}\ \\
Consider an $S[B]$ system. Given a $B$ state $q$ and an $S$ state $r$ such that $ q \models L(r)$, then $q$ is strong adaptable to $r$ if and only if the strong adaptation CTL formula (equation~\ref{eq:strong}) is true in $\mathcal{K}(S[B])$ at state $(q,r,\emptyset)$. 
Formally, given a state $q \in [[L(r)]]$ 
$$
q \ |_w \ r 
\iff (q,r,\emptyset) \models_{\mathrm{CTL}} \mathbf{AG}(( adapting \implies \mathbf{AF} \ steady) \; \wedge \; progress)
$$ 
\end{thm}
\begin{proof}
See~\ref{proof:thm2}
\end{proof}

\begin{cor}\ \\
Consider an $S[B]$ system. Then, $S[B]$ is strong adaptable if and only if 
$$
t_{0}  \models_{\mathrm{CTL}} \mathbf{AG}(( adapting \implies \mathbf{AF} \ steady) \; \wedge \; progress)
$$ 
where $t_0$ is  the initial state of $\mathcal{K}(S[B])$.
 \end{cor}
 \begin{proof}
 As in the weak case, the thesis follows easily from Definition~\ref{def:strongadapt} and from Theorem~\ref{teo:strong}.
 \end{proof}
 
Note that since we assume that the behavioural and the structural state machines are finite state, then the CTL adaptability properties can be model checked. This means that the defined notions of weak and strong adaptability are decidable and that the problem of adaptability checking can be reduced to a classical CTL model checking problem.

\subsection{State Space Dimension}

CTL model checking has been widely investigated in the literature and relies on efficient tools like NuSMV~\cite{cimatti2002nusmv}. The time computational complexity of the model checking problem for CTL is $\mathit{O}((n + m) \cdot |\psi|)$, where $n$ is the number of states in the Kripke structure, $m$ is the number of transitions and $|\psi|$ is the length of the formula, i.e.\ the number of operators in its parse tree. 

The well-known problem in this area is the so-called state explosion problem, that is to say the high number of states (and thus transitions) that comes out from even relatively short descriptions of systems composed of concurrent interactive components. It is not in the scope of this work to discuss and refer the high research efforts that are currently going on in this area. A good starting point can be found in \cite{katoen2008}. We will just give a brief estimation of the dimension of the state space given a certain $S[B]$ system, which is the dominant complexity factor, considering that our formulae for the adaptability checking have constant length 4. The computational complexity of adaptability checking is therefore $\mathit{O}(n + m)$.

Note that in our context the usual sources of state explosion (components, concurrency) are ``hidden'' inside the behavioural level $B$. This is because, as we pointed out in Section~\ref{sec:introadapt}, we want to work on the very basic model of computation of finite state machines and maintain a black-box view of the behavioural level from the structural level point of view. Thus, we take as the dominant dimension of the problem the ``already exploded'' number of states of the behavioural level. Then, we estimate what the definition of an $S[B]$ model adds up to this explosion. 

Recall that the Kripke structure to model check is $\mathcal{K}(S[B])$, derived from the flat semantics $\mathcal{F}(S[B])$. Depending on the dimension of $S$, the flat semantics could possibly lead to a transition system larger than the state space of the behavioural model $B$ since the states are formally tuples in $Q \times R \times (\{(\psi,r') \mid \exists r \in R . \ r \goes{\psi}_S r' \} \cup \{\emptyset \})$. The dimension $n$ of the state space is $\mathit{O}(|Q| \cdot 2^{|Q|} \cdot 2^{|Q|})$, where $|\cdot|$ is set cardinality, based only on the number of states in the behavioural level and due to the higher order nature of $S$.
However, considering the intended role of the structural level $S$, the number of $S$ states is never exponential. This is because the different $S$ states represents different modes of operations and, thus, usually stand for disjoint sets of $B$ states. For this reason, a more realistic estimation of the state space dimension $n$ should be expressed w.r.t.\ both $B$ and $S$ dimensions, yielding a number that is $\mathit{O}(|Q| \cdot |R| \cdot |\goes{}_{S}|)$.

\section{Discussion and Conclusion}
\label{sect:conclusion}
In this work we presented a formal hierarchical model for
multi-level self-adaptive systems, consisting of a lower behavioural level and an upper structural level. The $B$ level is a state machine describing the behaviour of the system and the $S$ level is a second-order state
machine accounting for the constraints which the system has to
comply with. $S$ states identify stable regions that the $B$ level may
reach by performing adaptation paths. 

The adaptation semantics of the
multi-level system is given by a flattened transition system that implements a top-down and behavioural adaptation model: adaptation starts whenever the current $B$ state does not meet the constraints specified by the current $S$ state. Then, adaptation towards a target $S$ state $r'$ ends successfully when the system ends up in a different $B$ state $q'$ such that $q'$ satisfies the constraints in $r'$.

We tackled the adaptability checking problem by firstly characterizing two degrees of adaptability: \textit{weak adaptability}, for verifying if the system is able to adapt successfully for some adaptation paths; and \textit{strong adaptability}, for verifying if the system is able to adapt successfully for all possible adaptation paths. Then, we defined weak and strong adaptation as relations over the set of $B$ states and the set of $S$ states, so that adaptability is verified when an appropriate adaptation relation can be built. We also provided a logical formulation of weak and strong adaptability, in terms of CTL formulae. Finally, by proving that the logical characterization is formally equivalent to the relational one, we demonstrated that the adaptability checking problem can be reduced to a classical model checking problem. We derived the computational complexity of the problem and showed that the state space dimension is polynomial in the dimension of the original $S[B]$ system.

The approach has been elucidated through an example of self-adaptive software systems: the motion controller of an Autonomous Transport Vehicle. We considered two structural levels: $S_0$, which supports both a normal operation mode and a fallback mode to which the system adapts in case of traffic congestion; and $S_1$, which only supports the normal operation mode. Keeping the behavioural model $B$ fixed, we derived the flat semantics of $S_0[B]$ and $S_1[B]$ and we compared their adaptation capabilities, showing that the former is strong adaptable, while the latter is only weak adaptable.

We report that this work gives a formal computational characterization
of self-adaptive systems, and a novel and well-grounded formulation of the concept of adaptability. We elaborate effective formal methods to investigate and solve the problem of adaptability checking. Provided that $S[B]$ systems are based on a general and essential model of computation (state machines), our results are general too and can be easily declined into richer and more expressive models.

\subsection{Multiple Levels and Modular Adaptability Checking}
Although we have investigated the relationships between two fundamental levels, it is possible to show how our model can easily scale-up to an arbitrary number of levels, arising from the composition of multiple $S[B]$ systems. We give an intuition of how a higher order $S[B]$ can be defined. In these settings, a first-order $S[B]$ systems is a ``classical'' system, as defined in Sect.~\ref{sect:model}. For $n>1$, a $n^{th}$-order $S[B]$ system is an $S[B]$ system $S^n[B^n]$, where $S^n = (R^n, r_0^n, A^n, \mathcal{O}^n, \goes{}^n_S, L^n)$ is the structural level; and $B^n = \|_{i \in I} \ \mathcal{F}(S^{n-1}[B^{n-1}]_i)$ is the behavioural level resulting from the application of a parallel composition operation `$\|$' to the flattened semantics of a family of $n-1^{th}$-order $S[B]$ systems, indexed by $i \in I$.

Further, due to the separation between the $S$ and the $B$ levels, modular techniques for adaptability checking could be exploited in our model. Let $S_1=(R_1, r_{1_0}, A_1, \mathcal{O}_1, \goes{}_{S_1}, L_1)$ and $S_2=(R_2, r_{2_0}, A_2, \mathcal{O}_2, \goes{}_{S_2}, L_2)$ be two $S$ levels. For instance, we would be interested in showing if the adaptation capabilities of $S_2$ are preserved by $S_1$, in the case that $S_1$ refines $S_2$, or $S_1\preceq S_2$. 
To our purposes we may assume that $S_1\preceq S_2$ iff a suitable simulation relation $\mathcal{R} \in R_1 \times R_2$ exists.

The following result would come quite straightforwardly: if $S_1\preceq S_2$, it can be shown that for every $B$ level, if $S_2[B]$ is strong adaptable, then $S_1[B]$ is strong adaptable too, i.e.\ that refinements at the $S$ level would preserve strong adaptability. On the other hand, refinements do not necessarily preserve weak adaptability when $S_2[B]$ is weak adaptable but not strong adaptable. Instead, we cannot make any assumption on the adaptability of $S_2$ based just on the adaptability of its refinement $S_1$.

Modular adaptability could be investigated also in the opposite case, i.e.\ making the $S$ level vary and considering two behavioural levels $B_1$ and $B_2$ such that $B_1\preceq B_2$. Intuitively, it can be demonstrated that in this case abstractions at the $B$ level preserve weak adaptation, or alternatively that if $S[B_1]$ is weak adaptable, then $S[B_2]$ is weak adaptable too for each structural level $S$.

We leave the two topics briefly introduced above as future work, where also other features of the $S[B]$ model can be developed.

\subsection{Related Work}
\smallskip\noindent\textit{Behavioural Adaptation}. Several efforts have been made in the formal modelling of self-adaptive
software. Zhang et al.\ give a general state-based model of
self-adaptive programs, where the adaptation process is seen as a transition
between different non-adaptive regions in the state space of the
program~\cite{zhang2006model}. In order to verify the correctness of adaptation
they define a logic called A-LTL (an adapt-operator extension to LTL) and
model-checking algorithms~\cite{zhang2009modular} for verifying adaptation
requirements. Similarly, in $S[B]$ systems adaptation can be seen as a transition in the $S$ level between two steady regions (the $S$ states), which corresponds to performing a path at the $B$ level. However, in our model the
steady-state regions are represented in a more declarative way using
constraints associated to the states of the $S$ level. 
Adaptation of the $B$ level is not
necessarily instantaneous and during this phase the system is left unconstrained
but an invariant condition that is required to be met during adaptation.
Differently to~\cite{zhang2006model}, the invariants are specific for every
adaptation transition making this process controllable in a finer way.

PobSAM~\cite{khakpour2011formal,khakpour2010formal} is another formal model for self-adaptive systems, where actors expressed in Rebeca are governed by managers that enforce dynamic policies
(described in an algebraic language) according to which actors adapt their
behaviour. Different adaptation modes allow to handle events occurring during
adaptation and ensuring that managers switch to a new configuration only once
the system reaches a safe state. Similarly to our structural and behavioural levels, the structure of a PobSAM model is built on multiple levels: the level of managed actors, the levels of autonomous managers and a view level, which acts as a sort of observation function over the state variables of the actors. Further, a recently published extension of PobSAM called HPobSAM~\cite{khakpour2012hpobsam}, enables the hierarchical refinement of managed components. On the other hand, $S[B]$ systems are based on the general formalism of state machines and enjoy the property that higher levels lead to higher-order structures: indeed, the $S$ level can be interpreted as a second-order $B$ level, since an $S$ state identifies a set of stable $B$ states and firing an $S$ transition means performing an adaptation path, i.e.\ a sequence of $B$ transitions.

In the position paper by Bruni et
al.~\cite{bruni2012conceptual}, adaptation is defined as the run-time
modification of the control data of a system and this approach is instantiated into a formal model based on labelled transition systems. They consider a system $S$ that is embedded in some environment $\mathcal{E}$ and that has to fulfil a goal $\psi$. When the environment and the goal are fixed, the system $S$ is such that:
$$\mathcal{E}[S] \models \psi,$$
where $\mathcal{E}[S]$ can be alternatively expressed as the parallel composition $\mathcal{E} \| S$. However, $S$ may operate under run-time modifications in the environment and goal. Thus, when the environment $\mathcal{E}$ changes into $\mathcal{E}'$ and the goal $\psi$ into $\psi'$, $S$ adapts itself into $S'$ such that:
$$\mathcal{E}'[S'] \models \psi'.$$
As shown in~\cite{uchitel2012}, the problem of finding such an $S'$ can be formulated as an LTS control problem.
Similarly, in our context we can see an $S[B]$ system as a behavioural model $B$ embedded in a structure $S$. Let $\bar{r}_S[\bar{q}_B]$ denote the current state of the $S[B]$ system. Then,
$$\langle S[B], \bar{r}_S[\bar{q}_B] \rangle \models \psi_{S[B]},$$
where $\psi_{S[B]} \equiv \bar{q}_B \in L(\bar{r}_S)$, i.e.\ the current $B$ state must meet the constraints imposed by the current $S$ state. Therefore, the structural level $S$ not only acts as the operating environment for $B$, but also encodes the goal $\psi$, which requires that $B$ has to move within the stable region identified by the constraints in the current $S$ state. We can imagine that whenever $\psi$ is no longer satisfied, adaptation produces a system $S'[B']$ from the current system $S[B]$, in a way that
$$\langle S'[B'], \bar{r}_{S'}[\bar{q}_{B'}] \rangle \models \psi_{S'[B']}.$$
In our settings, the control data component is not explicitly implemented and there are no transitions that can be directly controlled, but due to the top-down adaptation semantics, the $S$ level provides some kind of control mechanism on $B$. Indeed, during the steady phase, $S$ forbids any transition to a $B$ state that violates the current constraints if there is at least one transition to a state satisfying them. Instead, when adaptation starts, we can think that $S$ outputs to $B$ a list of target $S$ states, so directing the evolution of $B$ towards a set of possible goals and excluding those transitions that lead to states violating the adaptation invariant.

The adaptation as control data modification view of~\cite{bruni2012conceptual} is implemented also in the formalism of Adaptable Interface Automata (\textsc{aias})~\cite{bruni2013adaptable}, which extends Interface Automata~\cite{de2001interface} with state-labelling atomic propositions, a subset of which - the control propositions - models the control data. Adaptation occurs in correspondence of transitions that change the control propositions, and actions labelling such transitions are called control actions. Similarly to our $S[B]$ systems, an adaptation phase is thus a sequence of adaptation transitions. Let $\mathcal{A}^C, \mathcal{A}^I, \mathcal{A}^O$ denote the set of control actions and the sets of input and output actions of the underlying interface automaton, respectively. On top of these actions, the authors provides different characterizations of an \textsc{aia} $P$: $P$ is adaptable when $\mathcal{A}^C \neq \emptyset$; $P$ is controllable when $\mathcal{A}^C \cap \mathcal{A}^I \neq \emptyset$; and $P$ is self-adaptive when $\mathcal{A}^C \cap \mathcal{A}^O \neq \emptyset$. Given an \textsc{aia} $P$, adaptability properties are defined as those that are satisfied by $P$, but are not satisfied if control actions are removed, i.e.\ by the \textsc{aia} $P_{|\mathcal{A} \setminus \mathcal{A}^C }$. In addition, the authors show how our notions of weak and strong adaptability can be encoded in their framework. 


Theorem-proving techniques have also been used for assessing the
correctness of adaptation: in~\cite{kulkarni2004correctness} a proof lattice
called transitional invariant lattice is built to verify that an adaptive
program satisfies global invariants before and after adaptation. In particular
it is proved that if it is possible to build that lattice, then adaptation is
correct. Instead, in our model the notion of correctness of adaptation is formalized by means of the weak and strong adaptability relations, that can be alternatively expressed as CTL formulae, thus reducing the problem of adaptability checking to a classical model checking problem.

Furthermore, in~\cite{viroli2011spatial}, the authors define a spatial and chemical-inspired tuple-space model in the context of distributed pervasive services. They show that equipping the classical tuple-space model with reaction and diffusion rules makes possible to support features like adaptivity and competition among services, by implementing Lotka-Volterra-like rules. Similarly, $S[B]$ systems are inspired by complex natural systems, where the dichotomy between the behavioural level and the structural level may represent for instance the genotype and the phenotype level of an organism, respectively; or, by using the metaphor of multiscale systems, the $B$ level can be used to model the system at the micro-scale (e.g.\ cellular scale), while the $S$ level would represent the emergent macro-scale features (e.g.\ the tissue).
However, our model is not quantitative and is not based on a coordination model, but on a simple and general model of computation (state machines). Moreover, adaptation is not the result of nature-inspired rules, but is rigorously defined from the $B$  and the $S$ level by operational semantics rules.

\smallskip\noindent\textit{Structural Adaptation}. Our model currently supports only behavioural adaptation, but several approaches have been recently defined also in the context of structural adaptation, most of them relying on dynamic software architectures. In~\cite{bradbury2004survey}, several formal techniques for specifying self-managing architectures, i.e.\ able to support autonomous run-time architectural changes, are surveyed and compared.

An important line of research focuses on the application of graph-based methods, initiated by Le M{\'e}tayer's work~\cite{le1998describing} where architectural styles are captured by graph grammars and graph rewriting rules are used to enable architectural reconfiguration. Other relevant literature in this field includes the Architectural Design Rewriting (ADR) framework~\cite{bruni2008style}, in which an architectural style is described as an algebra over typed graphs (i.e.\ the architectures), whose operators correspond to term-rewriting rules. ADR supports the well-formed compositions of architectures, the hierarchical specification of styles, style checking and style-preserving reconfiguration. A method for selecting which rules to apply for maintaining a particular architectural style against unexpected run-time reconﬁgurations is proposed in~\cite{poyias2012}, where the authors extend ADR rules with pre- and post-conditions, representing invariants to be met by an architecture before and after the application of a rule. Tool support is discussed in~\cite{bruni2008graph}, where the authors provide an in-depth comparison between the implementations of typed graph grammars in Alloy~\cite{jackson2011software} and of ADR in Maude~\cite{clavel2002maude}. 

Moreover, in~\cite{becker2006symbolic}, an approach for verifying safety properties in structurally adapting multi-agent systems is presented. By modelling a system as a graph and its evolution by graph transformation rules, safety properties are verified by means of structural invariants, i.e.\ a set of forbidden graph patterns.

\smallskip\noindent\textit{Hierarchical and Multi-level Methods}. Besides $S[B]$ systems, multi-level approaches have been extensively used for the modelling of self-adaptive software systems. For instance,
in~\cite{corradini2006relating} Corradini et al.\ identify and formally relate
three different levels: the \textit{requirement level}, dealing with high-level
properties and goals; the \textit{architectural level}, focusing on the
component structure and interactions between components; and the
\textit{functional level}, accounting for the behaviour of a single component. 

Furthermore, Kramer and Magee~\cite{kramer2007self} define a three-level
architecture for self-managed systems consisting of a \textit{component control
level} that implements the functional behaviour of the system by means of
interconnected components; a \textit{change management level} responsible for
changing the lower component architecture according to the current status and
objectives; and a \textit{goal management level} that modifies the lower change
management plans according to high-level goals. 

Hierarchical finite state
machines, in particular Statecharts~\cite{harel1987statecharts} have also been employed to
describe the multiple architectural levels in self-adaptive software
systems~\cite{karsai2003approach,shin2005self}. 

Relevant applications of multi-level approaches include the work by Zhao et al.~\cite{zhao2011model}, where the authors present a two-level model for self-adaptive systems consisting of a functional behavioural level - accounting for the application logic and modelled as state machines - and an adaptation level, accounting for the adaptation logic and represented with a mode automata~\cite{maraninchi2003mode}. In this case, each mode in the adaptation level is associated with different functional state machines and adaptation is seen as a change of mode. Adaptation properties are described and checked by means of a mode-based extension of LTL called mLTL.

Another accepted fact is that higher levels in complex adaptive systems lead to
higher-order structures. Here the higher $S$ level is described by means of a
second order state machine (i.e.\ a state machine over the power set of the
$B$ states). Similar notions have been formalized by Baas~\cite{nils3baas} with
the \textit{hyperstructures} framework for multi-level and higher-order
dynamical systems; and by Ehresmann and Vanbremeersch with their \textit{memory
evolutive systems}~\cite{ehresmann2007memory}, a model for hierarchical
autonomous systems based on category theory.

There are several other works worth mentioning, but here we do not aim at
presenting an exhaustive state-of-the-art in this widening research field. We
address the interested reader to the
surveys~\cite{cheng2009software,salehie2009self} for a general introduction to
the essential aspects and challenges in the modelling of self-adaptive software
systems.

\section*{Acknowledgements}

The authors want to thank the anonymous reviewers for their valuable suggestions and prof.\ Mario Rasetti for the continuous inspiration and the useful discussions about the topics of this work and its general context. This work was partially supported by the project  ``TOPDRIM: \textit{Topology Driven Methods for Complex Systems}'' funded by the European Commission (FP7 ICT FET Proactive - Grant Agreement N.\ 318121).






\begin{thebibliography}{47}
\expandafter\ifx\csname natexlab\endcsname\relax\def\natexlab#1{#1}\fi
\providecommand{\bibinfo}[2]{#2}
\ifx\xfnm\relax \def\xfnm[#1]{\unskip,\space#1}\fi
\bibitem[{Baas(1994)}]{nils3baas}
\bibinfo{author}{N.~Baas}, \bibinfo{title}{Emergence, hierarchies, and
  hyperstructures}, in: \bibinfo{editor}{C.~Langton} (Ed.),
  \bibinfo{booktitle}{Artificial Life III}, volume~\bibinfo{volume}{17},
  \bibinfo{publisher}{Addison Wesley}, \bibinfo{year}{1994}, pp.
  \bibinfo{pages}{515--537}.
\bibitem[{Baier and Katoen(2008)}]{katoen2008}
\bibinfo{author}{C.~Baier}, \bibinfo{author}{J.P. Katoen},
  \bibinfo{title}{Principles of Model Checking}, \bibinfo{publisher}{The MIT
  Press}, \bibinfo{year}{2008}.
\bibitem[{Bartocci et~al.(2010{\natexlab{a}})Bartocci, Cacciagrano,
  Di~Berardini, Merelli and Tesei}]{shape2}
\bibinfo{author}{E.~Bartocci}, \bibinfo{author}{D.~Cacciagrano},
  \bibinfo{author}{M.~Di~Berardini}, \bibinfo{author}{E.~Merelli},
  \bibinfo{author}{L.~Tesei}, \bibinfo{title}{{Timed Operational Semantics and
  Well-Formedness of Shape Calculus}}, \bibinfo{journal}{Scientific Annals of
  Computer Science} \bibinfo{volume}{20} (\bibinfo{year}{2010}{\natexlab{a}})
  \bibinfo{pages}{33--52}.
\bibitem[{Bartocci et~al.(2010{\natexlab{b}})Bartocci, Corradini, Di~Berardini,
  Merelli and Tesei}]{shape1}
\bibinfo{author}{E.~Bartocci}, \bibinfo{author}{F.~Corradini},
  \bibinfo{author}{M.~Di~Berardini}, \bibinfo{author}{E.~Merelli},
  \bibinfo{author}{L.~Tesei}, \bibinfo{title}{{Shape Calculus. A Spatial Mobile
  Calculus for 3D Shapes}}, \bibinfo{journal}{Scientific Annals of Computer
  Science} \bibinfo{volume}{20} (\bibinfo{year}{2010}{\natexlab{b}})
  \bibinfo{pages}{1--31}.
\bibitem[{Becker et~al.(2006)Becker, Beyer, Giese, Klein and
  Schilling}]{becker2006symbolic}
\bibinfo{author}{B.~Becker}, \bibinfo{author}{D.~Beyer},
  \bibinfo{author}{H.~Giese}, \bibinfo{author}{F.~Klein},
  \bibinfo{author}{D.~Schilling}, \bibinfo{title}{Symbolic invariant
  verification for systems with dynamic structural adaptation}, in:
  \bibinfo{booktitle}{Proceedings of the 28th international conference on
  Software engineering}, ICSE '06, \bibinfo{publisher}{ACM},
  \bibinfo{year}{2006}, pp. \bibinfo{pages}{72--81}.
\bibitem[{Bouchachia and Nedjah(2012)}]{bouchachia2012introduction}
\bibinfo{author}{A.~Bouchachia}, \bibinfo{author}{N.~Nedjah},
  \bibinfo{title}{Introduction to the special section on self-adaptive systems:
  Models and algorithms}, \bibinfo{journal}{ACM Transactions on Autonomous and
  Adaptive Systems (TAAS)} \bibinfo{volume}{7} (\bibinfo{year}{2012})
  \bibinfo{pages}{13}.
\bibitem[{Bradbury et~al.(2004)Bradbury, Cordy, Dingel and
  Wermelinger}]{bradbury2004survey}
\bibinfo{author}{J.~Bradbury}, \bibinfo{author}{J.~Cordy},
  \bibinfo{author}{J.~Dingel}, \bibinfo{author}{M.~Wermelinger},
  \bibinfo{title}{A survey of self-management in dynamic software architecture
  specifications}, in: \bibinfo{booktitle}{Proceedings of the 1st ACM SIGSOFT
  workshop on Self-managed systems}, \bibinfo{organization}{ACM}, pp.
  \bibinfo{pages}{28--33}.
\bibitem[{Bruni et~al.(2008{\natexlab{a}})Bruni, Bucchiarone, Gnesi, Hirsch and
  Lafuente}]{bruni2008graph}
\bibinfo{author}{R.~Bruni}, \bibinfo{author}{A.~Bucchiarone},
  \bibinfo{author}{S.~Gnesi}, \bibinfo{author}{D.~Hirsch},
  \bibinfo{author}{A.L. Lafuente}, \bibinfo{title}{Graph-based design and
  analysis of dynamic software architectures}, in:
  \bibinfo{booktitle}{Concurrency, Graphs and Models}, volume
  \bibinfo{volume}{6065} of \textit{\bibinfo{series}{Lecture Notes in Computer
  Science}}, \bibinfo{publisher}{Springer}, \bibinfo{year}{2008}{\natexlab{a}},
  pp. \bibinfo{pages}{37--56}.
\bibitem[{Bruni et~al.(2013)Bruni, Corradini, Gadducci, Lafuente and
  Vandin}]{bruni2013adaptable}
\bibinfo{author}{R.~Bruni}, \bibinfo{author}{A.~Corradini},
  \bibinfo{author}{F.~Gadducci}, \bibinfo{author}{A.L. Lafuente},
  \bibinfo{author}{A.~Vandin}, \bibinfo{title}{Adaptable transition systems},
  in: \bibinfo{booktitle}{Recent Trends in Algebraic Development Techniques},
  volume \bibinfo{volume}{7841} of \textit{\bibinfo{series}{Lecture Notes in
  Computer Science}}, \bibinfo{publisher}{Springer}, \bibinfo{year}{2013}, pp.
  \bibinfo{pages}{95--110}.
\bibitem[{Bruni et~al.(2012)Bruni, Corradini, Gadducci, Lluch~Lafuente and
  Vandin}]{bruni2012conceptual}
\bibinfo{author}{R.~Bruni}, \bibinfo{author}{A.~Corradini},
  \bibinfo{author}{F.~Gadducci}, \bibinfo{author}{A.~Lluch~Lafuente},
  \bibinfo{author}{A.~Vandin}, \bibinfo{title}{A conceptual framework for
  adaptation}, in: \bibinfo{booktitle}{Fundamental Approaches to Software
  Engineering}, volume \bibinfo{volume}{7212} of
  \textit{\bibinfo{series}{Lecture Notes in Computer Science}},
  \bibinfo{publisher}{Springer}, \bibinfo{year}{2012}, pp.
  \bibinfo{pages}{240--254}.
\bibitem[{Bruni et~al.(2008{\natexlab{b}})Bruni, Lluch-Lafuente, Montanari and
  Tuosto}]{bruni2008style}
\bibinfo{author}{R.~Bruni}, \bibinfo{author}{A.~Lluch-Lafuente},
  \bibinfo{author}{U.~Montanari}, \bibinfo{author}{E.~Tuosto},
  \bibinfo{title}{Style-based architectural reconfigurations},
  \bibinfo{journal}{Bulletin of the European association for theoretical
  computer science} \bibinfo{volume}{94} (\bibinfo{year}{2008}{\natexlab{b}})
  \bibinfo{pages}{161--180}.
\bibitem[{Cheng et~al.(2009)Cheng, de~Lemos, Giese, Inverardi, Magee,
  Andersson, Becker, Bencomo, Brun, Cukic et~al.}]{cheng2009software}
\bibinfo{author}{B.~Cheng}, \bibinfo{author}{R.~de~Lemos},
  \bibinfo{author}{H.~Giese}, \bibinfo{author}{P.~Inverardi},
  \bibinfo{author}{J.~Magee}, \bibinfo{author}{J.~Andersson},
  \bibinfo{author}{B.~Becker}, \bibinfo{author}{N.~Bencomo},
  \bibinfo{author}{Y.~Brun}, \bibinfo{author}{B.~Cukic}, et~al.,
  \bibinfo{title}{Software engineering for self-adaptive systems: A research
  roadmap}, \bibinfo{journal}{Software Engineering for Self-Adaptive Systems}
  (\bibinfo{year}{2009}) \bibinfo{pages}{1--26}.
\bibitem[{Cimatti et~al.(2002)Cimatti, Clarke, Giunchiglia, Giunchiglia,
  Pistore, Roveri, Sebastiani and Tacchella}]{cimatti2002nusmv}
\bibinfo{author}{A.~Cimatti}, \bibinfo{author}{E.~Clarke},
  \bibinfo{author}{E.~Giunchiglia}, \bibinfo{author}{F.~Giunchiglia},
  \bibinfo{author}{M.~Pistore}, \bibinfo{author}{M.~Roveri},
  \bibinfo{author}{R.~Sebastiani}, \bibinfo{author}{A.~Tacchella},
  \bibinfo{title}{{NuSMV 2: An opensource tool for symbolic model checking}},
  in: \bibinfo{booktitle}{Computer Aided Verification},
  \bibinfo{organization}{Springer}, pp. \bibinfo{pages}{359--364}.
\bibitem[{Clarke et~al.(1986)Clarke, Emerson and Sistla}]{clarke1986automatic}
\bibinfo{author}{E.~Clarke}, \bibinfo{author}{E.~Emerson},
  \bibinfo{author}{A.~Sistla}, \bibinfo{title}{Automatic verification of
  finite-state concurrent systems using temporal logic specifications},
  \bibinfo{journal}{ACM Transactions on Programming Languages and Systems
  (TOPLAS)} \bibinfo{volume}{8} (\bibinfo{year}{1986})
  \bibinfo{pages}{244--263}.
\bibitem[{Clarke and Emerson(1981)}]{clarke81}
\bibinfo{author}{E.M. Clarke}, \bibinfo{author}{E.A. Emerson},
  \bibinfo{title}{Design and synthesis of synchronization skeletons using
  branching time temporal logic}, in: \bibinfo{booktitle}{Logic of Programs},
  number \bibinfo{number}{131} in \bibinfo{series}{Lecture Notes in Computer
  Science}, \bibinfo{publisher}{Springer-Verlag}, \bibinfo{year}{1981}, pp.
  \bibinfo{pages}{52--71}.
\bibitem[{Clarke et~al.(1995)Clarke, Grumberg, McMillan and
  Zhao}]{clarke1995efficient}
\bibinfo{author}{E.M. Clarke}, \bibinfo{author}{O.~Grumberg},
  \bibinfo{author}{K.L. McMillan}, \bibinfo{author}{X.~Zhao},
  \bibinfo{title}{Efficient generation of counterexamples and witnesses in
  symbolic model checking}, in: \bibinfo{booktitle}{Proceedings of the 32nd
  annual ACM/IEEE Design Automation Conference}, pp. \bibinfo{pages}{427--432}.
\bibitem[{Clarke et~al.(1999)Clarke, Grumberg and Peled}]{clarke1999model}
\bibinfo{author}{E.M. Clarke}, \bibinfo{author}{O.~Grumberg},
  \bibinfo{author}{D.A. Peled}, \bibinfo{title}{Model Checking},
  \bibinfo{publisher}{The MIT press}, \bibinfo{year}{1999}.
\bibitem[{Clavel et~al.(2002)Clavel, Dur{\'a}n, Eker, Lincoln, Mart{\i}-Oliet,
  Meseguer and Quesada}]{clavel2002maude}
\bibinfo{author}{M.~Clavel}, \bibinfo{author}{F.~Dur{\'a}n},
  \bibinfo{author}{S.~Eker}, \bibinfo{author}{P.~Lincoln},
  \bibinfo{author}{N.~Mart{\i}-Oliet}, \bibinfo{author}{J.~Meseguer},
  \bibinfo{author}{J.F. Quesada}, \bibinfo{title}{Maude: Specification and
  programming in rewriting logic}, \bibinfo{journal}{Theoretical Computer
  Science} \bibinfo{volume}{285} (\bibinfo{year}{2002})
  \bibinfo{pages}{187--243}.
\bibitem[{Corradini et~al.(2006)Corradini, Inverardi and
  Wolf}]{corradini2006relating}
\bibinfo{author}{F.~Corradini}, \bibinfo{author}{P.~Inverardi},
  \bibinfo{author}{A.~Wolf}, \bibinfo{title}{On relating functional
  specifications to architectural specifications: a case study},
  \bibinfo{journal}{Science of Computer Programming} \bibinfo{volume}{59}
  (\bibinfo{year}{2006}) \bibinfo{pages}{171--208}.
\bibitem[{De~Alfaro and Henzinger(2001)}]{de2001interface}
\bibinfo{author}{L.~De~Alfaro}, \bibinfo{author}{T.~Henzinger},
  \bibinfo{title}{Interface automata}, \bibinfo{journal}{ACM SIGSOFT Software
  Engineering Notes} \bibinfo{volume}{26} (\bibinfo{year}{2001})
  \bibinfo{pages}{109--120}.
\bibitem[{De~Nicola and Vaandrager(1990)}]{de1990action}
\bibinfo{author}{R.~De~Nicola}, \bibinfo{author}{F.~Vaandrager},
  \bibinfo{title}{Action versus state based logics for transition systems},
  \bibinfo{journal}{Semantics of Systems of Concurrent Processes}
  (\bibinfo{year}{1990}) \bibinfo{pages}{407--419}.
\bibitem[{D'Ippolito et~al.(2012)D'Ippolito, Braberman, Piterman and
  Uchitel}]{uchitel2012}
\bibinfo{author}{N.~D'Ippolito}, \bibinfo{author}{V.~Braberman},
  \bibinfo{author}{N.~Piterman}, \bibinfo{author}{S.~Uchitel},
  \bibinfo{title}{The modal transition system control problem}, in:
  \bibinfo{booktitle}{FM 2012: Formal Methods}, volume \bibinfo{volume}{7436}
  of \textit{\bibinfo{series}{Lecture Notes in Computer Science}},
  \bibinfo{year}{2012}, pp. \bibinfo{pages}{155--170}.
\bibitem[{Ehresmann and Vanbremeersch(2007)}]{ehresmann2007memory}
\bibinfo{author}{A.~Ehresmann}, \bibinfo{author}{J.~Vanbremeersch},
  \bibinfo{title}{Memory evolutive systems: hierarchy, emergence, cognition},
  volume~\bibinfo{volume}{4}, \bibinfo{publisher}{Elsevier Science},
  \bibinfo{year}{2007}.
\bibitem[{Feferman(1974)}]{feferman1974applications}
\bibinfo{author}{S.~Feferman}, \bibinfo{title}{Applications of many-sorted
  interpolation theorems}, in: \bibinfo{booktitle}{Proceedings of the Tarski
  Symposium (Proc. Sympos. Pure Math., Vol. XXV, Univ. of California, Berkeley,
  Calif., 1971)}, pp. \bibinfo{pages}{205--223}.
\bibitem[{Feiler et~al.(2006)Feiler, Gabriel, Goodenough, Linger, Longstaff,
  Kazman, Klein, Northrop, Schmidt and Sullivan}]{feiler2006ultra}
\bibinfo{author}{P.~Feiler}, \bibinfo{author}{R.~Gabriel},
  \bibinfo{author}{J.~Goodenough}, \bibinfo{author}{R.~Linger},
  \bibinfo{author}{T.~Longstaff}, \bibinfo{author}{R.~Kazman},
  \bibinfo{author}{M.~Klein}, \bibinfo{author}{L.~Northrop},
  \bibinfo{author}{D.~Schmidt}, \bibinfo{author}{K.~Sullivan},
  \bibinfo{title}{Ultra-large-scale systems: The software challenge of the
  future}, \bibinfo{journal}{Software Engineering Institute}
  (\bibinfo{year}{2006}).
\bibitem[{Harel(1987)}]{harel1987statecharts}
\bibinfo{author}{D.~Harel}, \bibinfo{title}{Statecharts: A visual formalism for
  complex systems}, \bibinfo{journal}{Science of computer programming}
  \bibinfo{volume}{8} (\bibinfo{year}{1987}) \bibinfo{pages}{231--274}.
\bibitem[{Jackson(2011)}]{jackson2011software}
\bibinfo{author}{D.~Jackson}, \bibinfo{title}{Software Abstractions: Logic,
  Language, and Anlysis}, \bibinfo{publisher}{The MIT Press},
  \bibinfo{year}{2011}.
\bibitem[{Karsai et~al.(2003)Karsai, Ledeczi, Sztipanovits, Peceli, Simon and
  Kovacshazy}]{karsai2003approach}
\bibinfo{author}{G.~Karsai}, \bibinfo{author}{A.~Ledeczi},
  \bibinfo{author}{J.~Sztipanovits}, \bibinfo{author}{G.~Peceli},
  \bibinfo{author}{G.~Simon}, \bibinfo{author}{T.~Kovacshazy},
  \bibinfo{title}{An approach to self-adaptive software based on supervisory
  control}, \bibinfo{journal}{Self-adaptive software: applications}
  (\bibinfo{year}{2003}) \bibinfo{pages}{77--92}.
\bibitem[{Khakpour et~al.(2012)Khakpour, Jalili, Sirjani, Goltz and
  Abolhasanzadeh}]{khakpour2012hpobsam}
\bibinfo{author}{N.~Khakpour}, \bibinfo{author}{S.~Jalili},
  \bibinfo{author}{M.~Sirjani}, \bibinfo{author}{U.~Goltz},
  \bibinfo{author}{B.~Abolhasanzadeh}, \bibinfo{title}{Hpobsam for modeling and
  analyzing it ecosystems--through a case study}, \bibinfo{journal}{Journal of
  Systems and Software} \bibinfo{volume}{85} (\bibinfo{year}{2012})
  \bibinfo{pages}{2770--2784}.
\bibitem[{Khakpour et~al.(2011)Khakpour, Jalili, Talcott, Sirjani and
  Mousavi}]{khakpour2011formal}
\bibinfo{author}{N.~Khakpour}, \bibinfo{author}{S.~Jalili},
  \bibinfo{author}{C.~Talcott}, \bibinfo{author}{M.~Sirjani},
  \bibinfo{author}{M.~Mousavi}, \bibinfo{title}{Formal modeling of evolving
  self-adaptive systems}, \bibinfo{journal}{Science of Computer Programming}
  \bibinfo{volume}{78} (\bibinfo{year}{2011}) \bibinfo{pages}{3--26}.
\bibitem[{Khakpour et~al.(2010)Khakpour, Khosravi, Sirjani and
  Jalili}]{khakpour2010formal}
\bibinfo{author}{N.~Khakpour}, \bibinfo{author}{R.~Khosravi},
  \bibinfo{author}{M.~Sirjani}, \bibinfo{author}{S.~Jalili},
  \bibinfo{title}{Formal analysis of policy-based self-adaptive systems}, in:
  \bibinfo{booktitle}{Proceedings of the 2010 ACM Symposium on Applied
  Computing}, \bibinfo{publisher}{ACM}, \bibinfo{year}{2010}, pp.
  \bibinfo{pages}{2536--2543}.
\bibitem[{Kramer and Magee(2007)}]{kramer2007self}
\bibinfo{author}{J.~Kramer}, \bibinfo{author}{J.~Magee},
  \bibinfo{title}{Self-managed systems: an architectural challenge}, in:
  \bibinfo{booktitle}{Future of Software Engineering, 2007. FOSE'07},
  \bibinfo{organization}{IEEE}, pp. \bibinfo{pages}{259--268}.
\bibitem[{Kripke(1963)}]{kripke63}
\bibinfo{author}{S.~Kripke}, \bibinfo{title}{Semantical considerations on modal
  logic}, \bibinfo{journal}{Acta Philosophica Fennica} \bibinfo{volume}{16}
  (\bibinfo{year}{1963}) \bibinfo{pages}{83--94}.
\bibitem[{Kulkarni and Biyani(2004)}]{kulkarni2004correctness}
\bibinfo{author}{S.~Kulkarni}, \bibinfo{author}{K.~Biyani},
  \bibinfo{title}{Correctness of component-based adaptation},
  \bibinfo{journal}{Component-Based Software Engineering}
  (\bibinfo{year}{2004}) \bibinfo{pages}{48--58}.
\bibitem[{Laddaga(1997)}]{laddaga}
\bibinfo{author}{R.~Laddaga}, \bibinfo{title}{Self-adaptive software},
  \bibinfo{type}{Technical Report} \bibinfo{number}{98-12}, DARPA BAA,
  \bibinfo{year}{1997}.
\bibitem[{Le~M{\'e}tayer(1998)}]{le1998describing}
\bibinfo{author}{D.~Le~M{\'e}tayer}, \bibinfo{title}{Describing software
  architecture styles using graph grammars}, \bibinfo{journal}{Software
  Engineering, IEEE Transactions on} \bibinfo{volume}{24}
  (\bibinfo{year}{1998}) \bibinfo{pages}{521--533}.
\bibitem[{Maraninchi and R{\'e}mond(2003)}]{maraninchi2003mode}
\bibinfo{author}{F.~Maraninchi}, \bibinfo{author}{Y.~R{\'e}mond},
  \bibinfo{title}{Mode-automata: a new domain-specific construct for the
  development of safe critical systems}, \bibinfo{journal}{Science of Computer
  Programming} \bibinfo{volume}{46} (\bibinfo{year}{2003})
  \bibinfo{pages}{219--254}.
\bibitem[{Merelli et~al.(2012)Merelli, Paoletti and Tesei}]{merelli2012}
\bibinfo{author}{E.~Merelli}, \bibinfo{author}{N.~Paoletti},
  \bibinfo{author}{L.~Tesei}, \bibinfo{title}{A multi-level model for
  self-adaptive systems}, \bibinfo{journal}{EPTCS} \bibinfo{volume}{91}
  (\bibinfo{year}{2012}) \bibinfo{pages}{112--126}. \bibinfo{note}{{\it
  Proceedings of FOCLASA '12}}.
\bibitem[{Pnueli(1977)}]{pnueli77}
\bibinfo{author}{A.~Pnueli}, \bibinfo{title}{The temporal logic of programs},
  in: \bibinfo{booktitle}{18th IEEE Symposium on Foundations of Computer
  Science (FOCS)}, \bibinfo{publisher}{IEEE Computer Society},
  \bibinfo{year}{1977}, pp. \bibinfo{pages}{46--67}.
\bibitem[{Poyias and Tuosto(2012)}]{poyias2012}
\bibinfo{author}{K.~Poyias}, \bibinfo{author}{E.~Tuosto},
  \bibinfo{title}{Enforcing architectural styles in presence of unexpected
  distributed reconfigurations}, \bibinfo{journal}{EPTCS} \bibinfo{volume}{104}
  (\bibinfo{year}{2012}) \bibinfo{pages}{67--82}. \bibinfo{note}{{\it Proc. of
  ICE 2012}}.
\bibitem[{Sagasti(1970)}]{sagasti1970conceptual}
\bibinfo{author}{F.~Sagasti}, \bibinfo{title}{A conceptual and taxonomic
  framework for the analysis of adaptive behavior}, \bibinfo{journal}{General
  systems} \bibinfo{volume}{15} (\bibinfo{year}{1970})
  \bibinfo{pages}{151--160}.
\bibitem[{Salehie and Tahvildari(2009)}]{salehie2009self}
\bibinfo{author}{M.~Salehie}, \bibinfo{author}{L.~Tahvildari},
  \bibinfo{title}{Self-adaptive software: Landscape and research challenges},
  \bibinfo{journal}{ACM Transactions on Autonomous and Adaptive Systems (TAAS)}
  \bibinfo{volume}{4} (\bibinfo{year}{2009}).
\bibitem[{Shin(2005)}]{shin2005self}
\bibinfo{author}{M.~Shin}, \bibinfo{title}{Self-healing components in robust
  software architecture for concurrent and distributed systems},
  \bibinfo{journal}{Science of Computer Programming} \bibinfo{volume}{57}
  (\bibinfo{year}{2005}) \bibinfo{pages}{27--44}.
\bibitem[{Viroli et~al.(2011)Viroli, Casadei, Montagna and
  Zambonelli}]{viroli2011spatial}
\bibinfo{author}{M.~Viroli}, \bibinfo{author}{M.~Casadei},
  \bibinfo{author}{S.~Montagna}, \bibinfo{author}{F.~Zambonelli},
  \bibinfo{title}{Spatial coordination of pervasive services through
  chemical-inspired tuple spaces}, \bibinfo{journal}{ACM Transactions on
  Autonomous and Adaptive Systems} \bibinfo{volume}{6} (\bibinfo{year}{2011}).
\bibitem[{Zhang and Cheng(2006)}]{zhang2006model}
\bibinfo{author}{J.~Zhang}, \bibinfo{author}{B.~Cheng},
  \bibinfo{title}{Model-based development of dynamically adaptive software},
  in: \bibinfo{booktitle}{Proceedings of the 28th international conference on
  Software engineering}, \bibinfo{publisher}{ACM}, \bibinfo{year}{2006}, pp.
  \bibinfo{pages}{371--380}.
\bibitem[{Zhang et~al.(2009)Zhang, Goldsby and Cheng}]{zhang2009modular}
\bibinfo{author}{J.~Zhang}, \bibinfo{author}{H.~Goldsby},
  \bibinfo{author}{B.~Cheng}, \bibinfo{title}{Modular verification of
  dynamically adaptive systems}, in: \bibinfo{booktitle}{Proceedings of the 8th
  ACM international conference on Aspect-oriented software development},
  \bibinfo{publisher}{ACM}, \bibinfo{year}{2009}, pp.
  \bibinfo{pages}{161--172}.
\bibitem[{Zhao et~al.(2011)Zhao, Ma, Li and Li}]{zhao2011model}
\bibinfo{author}{Y.~Zhao}, \bibinfo{author}{D.~Ma}, \bibinfo{author}{J.~Li},
  \bibinfo{author}{Z.~Li}, \bibinfo{title}{Model checking of adaptive programs
  with mode-extended linear temporal logic}, in:
  \bibinfo{booktitle}{Engineering of Autonomic and Autonomous Systems (EASe),
  2011 8th IEEE International Conference and Workshops on},
  \bibinfo{organization}{IEEE}, pp. \bibinfo{pages}{40--48}.

\end{thebibliography}

\begin{thebibliography}{2}
\bibitem[1]{bartocci2012}
E.~Bartocci, P.~Li\`{o}, E.~Merelli, and N.~Paoletti,
\newblock{Multiple verification in complex biological systems: The bone
  remodelling case study.}
\newblock {\em Transactions on
  Computational Systems Biology XIV}, LNCS 7625, pp. 53--76, 2012.
  
  \bibitem[2]{paoletti2012multilevel}
N.~Paoletti, P.~Li\`{o}, E.~Merelli, and M.~Viceconti,
\newblock{Multilevel computational modeling and quantitative analysis
  of bone remodeling.}
\newblock {\em IEEE/ACM Transactions on Computational
  Biology and Bioinformatics}, 9, pp. 1366--1378, 2012.

\end{thebibliography}





\appendix
\section{Proofs}\label{appendix:proofs}
\subsection{Proposition 1 (Properties of flat semantics)}\label{proof:prop1}
\begin{proof}
~\begin{itemize}

\item[{\it (i + ii)}] Both the couples of rules \textsc{Steady + AdaptStart} and \textsc{Steady + AdaptStartEnd} ensure that there cannot exist a non-adapting state with both an outgoing non-adapting transition $\goes{r}$ and an outgoing adapting transition $\goes{r,\psi,r'}$. Indeed, the premises of the two rules, in both cases, are mutually exclusive by the fact that $(q \goes{}_{B} q' \wedge q' \models L(r))$ is the negation of $\forall q''.(q \goes{}_Bq'' \implies q'' \not\models L(r))$.

\item[{\it (iii)}] Rules \textsc{Adapt} and  \textsc{AdaptEnd} are the only ones producing an outgoing transition from an adapting state and none of them produces an $r$-labelled transition.

\item[{\it (iv)}] $(i) \wedge (ii) \wedge (iii) \implies (iv)$.

\item[{\it (v)}] Rule \textsc{Adapt} ensures that an adaptation transition is taken only if there are no other transitions that directly lead to the target $S$ state. Indeed \textsc{Adapt} and \textsc{AdaptEnd} are mutually exclusive, thus avoiding adaptation steps to be taken when adaptation can end. This also holds for the successful adaptation paths (of length 1) obtained with the rule \textsc{AdaptStartEnd}, whose premises are not compatible with those of rule \textsc{Adapt}.

\item[{\it (vi)}] 
Let $\pi$ be a generic path of $\mathcal{F}(S[B])$ starting at a state $(q,r,\emptyset)$. Let $i\geq 0$ be a generic position in $\pi$, denoted $\pi[i]$, such that $\pi[i] = (q_{i},r_{i},\emptyset)$. 
Then, by properties (i)-(iv), there are only the following cases:
\begin{enumerate}
  \item if $(q_{i},r_{i},\emptyset) \not \! \goes{r} \wedge \; (q_{i},r_{i},\emptyset) \;\;\;\;  \not \!\!\!\!\!\!\!\!\!\! \goes{r_{i}, \psi, r'}$ then the path stops at position $\pi[i] = (q_{i},r_{i},\emptyset)$;
\item if $(q_{i},r_{i},\emptyset) \goes{r}$, then by Rule \textsc{Steady} then the path will continue as $(q_{i},r_{i},\emptyset) \goes{r_{i}} (q_{i+1},r_{i+1},\emptyset)$. Thus, in this case $m_{i}=1$ and $n_{i} = 0$ and the path may continue after position $\pi[i+1]$;
\item if $(q_{i},r_{i},\emptyset) \goes{r_{i}, \psi, r'} (q_{i+1},r_{i+1},\emptyset)$ then Rule \textsc{AdaptStartEnd} has been applied and in this case $m_{i}=0$ and $n_{i} = 1$. The path may continue after position $\pi[i+1]$;
\item if $(q_{i},r_{i},\emptyset) (\goes{r_{i},\psi_{i},r_{i+1}})^{n_i}
 (q_{i+1},r_{i+1},\emptyset)$ then Rule \textsc{AdaptStart} has been applied, followed by zero or more applications of Rule \textsc{Adapt} and ended by the application of Rule \textsc{AdaptEnd}. In this case $m_{i}=0$ and $n_{i}  > 0$. The path may continue after position $\pi[i+1]$;
\item if $(q_{i},r_{i},\emptyset) (\goes{r_{i},\psi_{i},r_{i+1}})^{n_i}
 (q'',r_i,\rho) \goes{r_{i},\psi,r'} (q',r_{i}, \{(\psi, r')\})$ then Rule \textsc{AdaptStart} has been applied, followed by zero or more applications of Rule \textsc{Adapt} and, after $n_{i} \geq 0$ steps, the path has stopped because neither Rule \textsc{Adapt} nor Rule\textsc{AdaptEnd} could be applied. In this case the path stops and it does not reach a position $\pi[i+1]$ of the form $(q_{i+1},r_{i+1},\emptyset)$.
\end{enumerate}

If $\pi$ is finite then cases 3,4 or 5 occur for a certain number of steps, say $k-1 \geq 0$. At the $k$-th step, if the path stops because case 1 occurs, then it is of kind $(1)$. Otherwise, if it stops because case 5 occurs, then it is of kind $(2)$. If $\pi$ is infinite then it must be of kind $(1)$ because cases 1 and 5 can never occur.

\item[{\it (vii)}] 
By property $(vi)$ it follows that the positions $i$ in which $\pi[i] = (q_{i}, r_{i}, \emptyset)$ are those and only those in which the $S[B]$ system is in a steady state, i.e.\ it is either the first state $f_{0}$ in which by definition it holds $q_{0} \models L(r_{0})$, or it has been reached by using Rules \textsc{Steady}, \textsc{AdaptEnd} or \textsc{StartAdaptEnd}, which all explicitly check that $q_{i} \models L(r_{i})$. 
\end{itemize}
\end{proof}
\subsection{Proposition 2 (Union of Weak Adaptation Relations)}\label{proof:prop2}
\begin{proof}
If $(q,r) \in \mathcal{R}_{1} \cup \mathcal{R}_{2}$ then $(q,r) \in \mathcal{R}_{1}$ or $(q,r) \in \mathcal{R}_{2}$. Then it is possible to trivially verify all the conditions on $(q,r)$ of Definition~\ref{def:weakadapt} using the same proofs already available for $\mathcal{R}_{1}$ and $\mathcal{R}_{2}$, respectively, by substituting, in these proofs, $\mathcal{R}_{1}$ and $\mathcal{R}_{2}$ with $\mathcal{R}_{1} \cup \mathcal{R}_{2}$.
\end{proof}

\subsection{Lemma 1 (Propagation of Weak Adaptation Relation)}\label{proof:lemma1}
\begin{proof}
If the path exists then, by property {\it(vi)} of Proposition~1, its form is of the kind $(1)$, which equals to the one of the thesis. It remains to show the existence and that the weak adaptability propagates, i.e.\ $\forall i \geq 0 \; . \; q_{i} |_{w} r_{i}$. We use induction to construct the infinite path and to show the propagation.
When $i=0$, from $(q_{0} = q, r_{0} = r, \emptyset)$ we can construct the empty path and, by hypothesis, $q_{0} = q |_{w} r_{0} = r$. Moreover, the condition $\textsc{Progress}(q_{0}, r_{0})$ guarantees that the path can continue. At the generic step $i > 0$ suppose by induction that $q_{i} |_{w} r_{i}$. Thus, there exists a weak adaptation relation $\mathcal{R}_{i}$ containing $(q_{i},r_{i})$. In addition, the condition $\textsc{Progress}(q_{i}, r_{i})$ guarantees that the path can continue. This implies the existence of transition(s) $(q_i,r_i,\emptyset)
 (\goes{r_{i}})^{m_i}(\goes{r_{i},\psi_{i},r_{i+1}})^{n_i}
 (q_{i+1},r_{i+1},\emptyset)$. There are two cases:
\begin{itemize}
\item $m_{i} =1$. Thus $(q_i,r_i,\emptyset) \goes{r_{i}}  (q',r_{i+1} = r_{i}, \emptyset)$ for some $q' \in B$.  According to Definition~\ref{def:weakadapt} there is at least one $B$ state $q_{i+1}$ such that $\mathcal{R}_{i}$ contains $(q_{i+1},r_{i+1} = r_{i})$. Thus we can choose the transition with target $(q'=q_{i+1},r_{i+1} = r_{i}, \emptyset)$ and have $q_{i+1} |_{w} r_{i+1}$.
\item $m_{i} =0$. Thus $(q_i,r_i,\emptyset) \goes{r_{i}, \psi'_{i}, r'_{i+1}}$ for some $\psi'_{i}$ and $r'_{i+1}$. This is a transition leaving from a state $q_{i}$ that is weak adaptive to $r_{i}$. Thus, by the definition of weak adaptation, there exists a state $(q_{i+1},r_{i+1}, \emptyset)$ such that 
$(q_i,r_i,\emptyset) (\goes{r_{i}, \psi_{i}, r_{i+1}})^{n_{i}}  (q_{i+1},r_{i+1}, \emptyset)$ for some $n_{i} > 0$ and 
$(q_{i+1},r_{i+1}) \in \mathcal{R}_{i}$, i.e., $q_{i+1} |_{w} r_{i+1}$.
\end{itemize} 
\end{proof}
\subsection{Proposition 4 (Strong Adaptation implies Weak Adaptation)}\label{proof:prop4}
\begin{proof}
Since $q |_{s} r$ then there exists a strong adaptation relation $\mathcal{R}$ such that $(q,r) \in \mathcal{R}$. We construct a weak adaptation relation $\mathcal{R}'$ containing $(q,r)$, hence $q |_{w} r$. At the beginning we put $(q,r)$ in $\mathcal{R}'$, then for some transition $(q,r,\emptyset) \goes{r} (q',r,\emptyset)$ we add $(q,r')$, which belongs to $\mathcal{R}$, also to $\mathcal{R}'$. If no $\goes{r}$ transitions are possible then, if $(q,r,\emptyset) \goes{r,\psi,r'}$, since $q$ is strong adaptable to $r$, we select one of the surely existing successor states $(q',r',\emptyset)$ such that $(q',r') \in \mathcal{R}$ and we add $(q',r')$ to $\mathcal{R}'$. Then, we iterate this process for each new pair added in $\mathcal{R}'$. The process will terminate since the states are finite and the resulting $\mathcal{R}'$ will be, by construction, a weak adaptation relation.
\end{proof}

\subsection{Lemma 2 (Propagation of Strong Adaptation Relation)}\label{proof:lemma2}
\begin{proof}
Consider any generic path $\pi$ in $\mathcal{F}(S[B])$ starting from $(q,r,\emptyset)$. By Definition~\ref{def:strongadapt} the non-termination property $\textsc{Progress}(q_{i},r_{i})$ must hold for each state $(q_{i},r_{i},\emptyset)$, thus implying that the path is infinite. Thus, by property {\it(vi)} of Proposition~\ref{prop:semantics}, $\pi$ must be of the form:
$$
\begin{array}{l}
\pi = (q= q_0, r=r_0,\emptyset)
(\goes{r_{0}})^{m_0}(\goes{r_{0},\psi_{0},r_{1}})^{n_0} \cdots \\
\cdots (q_i,r_i,\emptyset)
 (\goes{r_{i}})^{m_i}(\goes{r_{i},\psi_{i},r_{i+1}})^{n_i}
 (q_{i+1},r_{i+1},\emptyset) \cdots 
\end{array}
$$
where for each $i$, either $m_i = 1 \wedge n_i = 0$ (steady transition) or $m_i = 0 \wedge n_i > 0$ (adaptation path). We prove inductively that $q_{i} |_{s} r_{i}$ for any $i \geq 0$. If $i=0$, then $q_{0}=q$ and $r_{0}=r$, thus the thesis is trivially true. If $i > 0$ suppose by induction that $q_{i} |_{s} r_{i}$. Thus, there exists a strong adaptation relation $\mathcal{R}_{i}$ containing $(q_{i},r_{i})$.  Consider the transition(s) $(q_i,r_i,\emptyset)
 (\goes{r_{i}})^{m_i}(\goes{r_{i},\psi_{i},r_{i+1}})^{n_i}
 (q_{i+1},r_{i+1},\emptyset)$. There are two cases:
\begin{itemize}
\item $m_{i} =1$. Thus $(q_i,r_i,\emptyset) \goes{r_{i}}  (q_{i+1},r_{i+1} = r_{i}, \emptyset)$ for several $B$ states playing the role of $q_{i+1}$. By the definition of strong adaptation, the relation $\mathcal{R}_{i}$ containing $(q_{i},r_{i})$ contains also $(q_{i+1},r_{i+1} = r_{i})$ for any state $q_{i+1}$, thus $q_{i+1} |_{s} r_{i+1}$.
\item $m_{i} =0$. Thus $(q_i,r_i,\emptyset) (\goes{r_{i}, \psi_{i}, r_{i+1}})^{n_{i}}  (q_{i+1},r_{i+1}, \emptyset)$ for some $n_{i} > 0$. Again by the definition of strong adaptation, this is a path leaving from a state $q_{i}$ that is strong adaptive to $r_{i}$. Thus, the reached state $(q_{i+1},r_{i+1}, \emptyset)$ must be such that $(q_{i+1},r_{i+1}) \in \mathcal{R}_{i}$, i.e., $q_{i+1} |_{s} r_{i+1}$.
\end{itemize} 
\end{proof}

\subsection{Proposition 5 (Construction of Strong Adaptation Relation)}\label{proof:prop5}
\begin{proof}
If $S[B]$ is strong adaptable then $f_{0} = (q_{0}, r_{0}, \emptyset)$ and $q_{0} |_{s} r_{0}$. By applying the propagation lemma of strong adaptation (Lemma~\ref{lem:propagation}), we get that all states $(q,r,\emptyset) \in Post^{*}(f_{0})$ are such that $q |_{s} r$. Thus, we can use $\mathcal{R}_{q,r}$ to denote the strong adaptation relation, containing $(q,r)$, that exists for each $(q,r,\emptyset) \in Post^{*}(f_{0})$. Moreover, we naturally deduce that, for each such pair, $q \models L(r)$. Thus, if we take $\hat{\mathcal{R}} = \bigcup_{(q,r,\emptyset) \in Post^{*}(f_{0})} \mathcal{R}_{q,r}$ we have, by definition of $\hat{\mathcal{R}}$ and of $\mathcal{R}$, $\mathcal{R} \subseteq \hat{\mathcal{R}}$. But $\mathcal{R}$ contains, by its definition, each possible pair $(q,r)$ such that $(q,r,\emptyset)$ is reachable from $f_{0}$ and by the rule \textsc{Steady} of the operational semantics, $q \models L(r)$. Thus, it must also hold  $\hat{\mathcal{R}} \subseteq \mathcal{R}$. Hence, $\mathcal{R} = \hat{\mathcal{R}}$. By Proposition~\ref{prop:unionstrong}, $\hat{\mathcal{R}} = \mathcal{R}$ is a strong adaptation relation.

For the converse, trivially, if $\mathcal{R} = \{ (q,r) \in Q \times R \mid (q,r,\emptyset) \in Post^{*}(f_{0}) \}$ is a strong adaptation relation then $S[B]$ is strong adaptable, because, by definition,  $(q_0,r_0) \in \mathcal{R}$.

\end{proof}

\subsection{Theorem 1 (Weak adaptability checking)}\label{proof:thm1}
\begin{proof}
$(\Rightarrow)$ Having $q \ |_w \ r $, by the propagation of weak adaptation (Lemma~\ref{lem:propagation_weak}) we can construct an infinite path $\pi$ in $\mathcal{F}(S[B])$, and thus in $\mathcal{K}(S[B])$, starting from $(q,r,\emptyset)$, which has the form specified in the Lemma and such that $q_{i} \ |_{w} \ r_{i}$ for all $i$. Such a path can be used to show that the given CTL formula is true. Since $q_{i} \ |_{w} \ r_{i}$, along $\pi$ the $progress$ proposition is true in all states. Moreover, whenever $m_{i} = 0$ and $n_{i} > 0$ in $(q_i,r_i,\emptyset) (\goes{r_{i}})^{m_i}(\goes{r_{i},\psi_{i},r_{i+1}})^{n_i}
 (q_{i+1},r_{i+1},\emptyset)$, the proposition $adapting$ is true in state $(q_i,r_i,\emptyset)$. In this case we reach, by following the path, the state $(q_{i+1},r_{i+1}, \emptyset)$ in which, since $progress$ is true, then also $steady$ is true. Note that $adapting$ is true also in all the intermediate states between $(q_{i},r_{i},\emptyset)$ and $(q_{i+1},r_{i+1},\emptyset)$. Following the same path, from all these states the same target state $(q_{i+1},r_{i+1},\emptyset)$, in which $steady$ is true, is reached. The cases in which $m_{i} = 1$ and $n_{i}=0$ correspond to states in which $adapting$ is false (by definition of $adapting$ and by property $(i)$ of Proposition~\ref{prop:semantics}), thus in this case the implication $(adapting \implies \mathbf{EF} \ steady)$ is vacuously true.

$(\Leftarrow)$ If the formula is true at state $(q,r,\emptyset)$, by definition of the semantics of CTL we have that there exists an infinite path $\pi$ in $\mathcal{K}(S[B])$ in which every state satisfies $progress$ and the sub-formula $(adapting \implies \mathbf{EF} \ steady)$. Such a path is a witness of the truth of the formula and can be calculated by a model checker usually in the form of a prefix followed by a cycle in which some reasonable fairness constraints hold\footnote{For a detailed discussion on the fairness constraints in CTL and on the generation of witnesses we refers to \cite{katoen2008,clarke1995efficient}.}. To show that $q \ |_{w} \ r$ we use such a path $\pi$ to generate a weak adaptation relation $\mathcal{R}_{\pi}$ as follows:
$$
\mathcal{R}_{\pi} = \{(q_{i},r_{i}) \mid \exists i \geq 0 \colon \pi[i] = (q_{i},r_{i},\emptyset)  \}
$$
First, we note that $(q= q_{0}, r= r_{0})$ is in $\mathcal{R}_{\pi}$ because $\pi[0] = (q=q_{0},r=r_{0},\emptyset)$. Then, we conclude the proof by showing, in the following, that $\mathcal{R}_{\pi}$ is indeed a weak adaptation relation. Let $i\geq 0$ and let $\pi[i] = (q_{i},r_{i},\emptyset)$. We check that for the generic pair $(q_{i},r_{i}) \in \mathcal{R}_{\pi}$ all the conditions of the weak adaptability definition (Def.~\ref{def:weakadapt}) hold:
\begin{itemize}
\item[(i)] $q_{i} \models L(r_{i})$ holds by property $(vii)$ of Proposition~\ref{prop:semantics}; $\textsc{Progress}(q_{i},r_{i})$ also holds because $\pi[i] = (q_{i},r_{i},\emptyset)$ is a state along an infinite path, thus it does not stop;
\item[(ii)] if in $\pi[i]$ we have that $(q_{i},r_{i},\emptyset) \goes{r}$ then we can take the transition $(q_{i},r_{i},\emptyset) \goes{r} (q_{i+1},r_{i+1},\emptyset)$ of $\pi$ and then $(q_{i+1} , r_{i+1}) \in \mathcal{R}_{\pi}$;
\item[(iii)] if in $\pi[i]$ the $(q_{i},r_{i},\emptyset) \goes{r,\psi,r'}$ for some $\psi$ and $r'$, then $\pi[i] \models_{\mathrm{CTL}} adapting$ and $\pi[i] \models_{\mathrm{CTL}} steady$. Thus, the sub-formula $(adapting \implies \mathbf{EF}$ $steady)$ is immediately true in $\pi[i]$. However, we know that $\pi$ continues and by property $(v)$ of Proposition~\ref{prop:semantics} we know that the semantics imposes that $\pi$ adapts as soon as possible. Thus, there are two further sub-cases:

\begin{itemize}
\item  the adaptation is immediate: $(q_{i},r_{i},\emptyset) \goes{r,\psi,r'} (q_{i+1} ,  r'= r_{i+1}, \emptyset)$ and thus $(q_{i+1} , r_{i+1}) \in \mathcal{R}_{\pi}$;
\item the adaptation cannot be immediate, thus in $\pi$ we have\\
$(q_{i},r_{i},\emptyset) \goes{r,\psi,r'} (q' , r_{i} , \{(\psi, r')\})$ for some $q' \in Q, \ \psi \in \Psi(\Sigma,A)$ and $r' \in R$. Again by the progress of $\pi$ and by the definition of $\mathcal{K}(S[B])$, it holds that $(q' , r_{i} , \{(\psi, r')\}) \models_{\mathrm{CTL}} adapting$ and that $(q' , r_{i} , \{(\psi, r')\}) \not \models_{\mathrm{CTL}} steady$. Now, the sub-formula $(adapting \implies \mathbf{EF} \ steady)$ is not immediately true and, by hypothesis, it holds in $(q' , r_{i} , \{(\psi, r')\})$. Thus, there exists in $\mathcal{K}(S[B])$ a path starting from $(q' , r_{i} , \{(\psi, r')\})$ and leading in $j$ steps, $j > 0$, to a state in which $steady$ holds, that is of the form $(q'', r', \emptyset)$. Among possibly others, the continuation of $\pi$ until the next state of the form $(q'' = q_{i+1} ,  r'= r_{i+1}, \emptyset)$ is a finite path satisfying this condition. Indeed, if it were not, then $\pi$ would either stop or infinitely continue along states that satisfy $adapting$ but not $steady$. This contradicts the fact that $\pi$ is a witness of the truth of the original CTL formula. Thus, we have that  $(q_{i},r_{i},\emptyset) (\goes{r_{i},\psi_{i},r_{i+1}})^{j+1}
 (q_{i+1},r_{i+1},\emptyset)$ (where $r_{i+1} = r'$ and $\psi_{i} = \psi$) and $(q_{i+1} , r_{i+1}) \in \mathcal{R}_{\pi}$.
\end{itemize}
\end{itemize}

\end{proof}

\subsection{Theorem 2 (Strong adaptability checking)}\label{proof:thm2}
\begin{proof}
$(\Rightarrow)$ Having $q \ |_s \ r $ we consider all paths $\pi$ in $\mathcal{K}(S[B])$ starting from $(q,r,\emptyset)$, which have the form and the properties stated in the proof of Lemma~\ref{lem:propagation}. All such paths can be used to show that the given CTL formula is true by using the same argument used in the  proof part $(\Rightarrow)$ of Theorem~\ref{teo:weak}, by turning the initial existential quantification into a universal one.

$(\Leftarrow)$ Also in this case the proof is similar to the proof part $(\Leftarrow)$ of Theorem~\ref{teo:weak}. However, the reasoning should be repeated not considering the witness of the truth of the formula, but a generic path $\pi$ starting from $(q,r,\emptyset)$. Moreover, the strong adaptation relation $\mathcal{R}$ must be defined generalising on all paths:
$$
\mathcal{R} = \{(q_{i},r_{i}) \mid \exists \pi \in \mathit{Paths}((q,r,\emptyset)) \colon \exists i \geq 0 \colon \pi[i] = (q_{i},r_{i},\emptyset)  \}
$$
Then, the checking of the conditions of the strong adaptability definition (Def.~\ref{def:strongadapt}) on every pair of $\mathcal{R}$ requires similar arguments to the proof of Theorem~\ref{teo:weak}.
\end{proof}
\section{Bone Remodelling Case Study}\label{appendix:bone}
\subsection{$S[B]$ Model}
Here a biological example is introduced: the bone remodelling (BR) process, which is intrinsically self-adaptive. The model presented in this paper is a simplified version of previous works by some of the co-authors in the formal computational modelling of BR. Here, aspects like molecular signalling, spatial location of cells in the bone tissue and quantitative dynamics have been omitted. For further details, we refer the interested reader to the papers~\cite{bartocci2012,paoletti2012multilevel}.

Bone remodelling is a process by which aged bone is continuously renewed in a balanced alternation of bone resorption, performed by cells called \textit{osteoclasts (Oc)}, and formation, performed by \textit{osteoblasts (Ob)}. It is responsible for repairing micro-damages, for maintaining mineral homeostasis and for the structural adaptation of bone in response to mechanical stress. Another kind of cells called \textit{osteocytes (Oy)} are responsible for the initiation of the remodelling process, by sending biomechanical signals that activate the resorption phase. Osteocytes act as mechanosensors: mechanical-induced signals from the tissue level are transmitted at cellular level, so that the intensity of the remodelling activity is regulated by the intensity of osteocytes' signalling (mechanotransduction). In this way, regular osteocytes' signalling leads to a regular remodelling activity, while events like a stronger mechanical stress or a micro-fracture induce a higher osteocytes' signalling and in turn to a more prominent remodelling activity.

This biological process is implemented as an $S[B]$ system where the $B$ level models the behaviour at the cellular level, while the $S$ level models the different phases of remodelling (initiation, resorption and formation). In particular, we consider two different $S$ levels: $S_0$, which describes the adaptation phases during a regular remodelling activity; and $S_1$ which extends $S_0$ in order to support also unexpected events (e.g.\ a micro-fracture), expressed by means of an over-signalling by osteocytes.

\subsubsection{Behavioural Level}
We define the following set of observables and associated sorts over the behavioural level (depicted in Fig.~\ref{fig:b-level_br}):
$$(Oc: \lbrace 0, 1, 2\rbrace, Ob: \lbrace 0, 1, 2, 3, 4\rbrace, Oy: \lbrace 0, 1, 2\rbrace)$$
Sorts and other symbols used hereafter are interpreted over the integers. $Oy$ is the variable modelling the number of active osteocytes, thus giving a measure of the strength of their signalling. $Oc$ models the availability of active osteoclasts, while $Ob$ is the variable accounting for the number of active osteoblasts. Note that these are underestimated values of bone cell abundances in a remodelling unit: realistic (approximate) ranges are: $[0,10]$ for osteoclasts, $[0,50]$ for osteoblasts and of $[0,9500]$ for osteocytes.

In order to avoid the exhaustive listing of all the $B$ transitions, we describe them by a set of guarded rules listed in Table~\ref{tbl:BRrules}, of the form
$$\text{\textsc{RuleName:}} \quad guard \goes{} update,$$
where \textsc{RuleName} is the name of the rule, $guard$ is a pre-condition indicating when the rule can be applied, thus determining the source states, and $update$ possibly assigns different values to observable variables, thus determining the target states.

Rule \textsc{Init} describes the initiation of bone remodelling from the state of quiescence and corresponds to the transition $(0,0,0)\goes{}(0,0,1)$. Rule \textsc{Oy$_+$} tells that if there is no signalling activity by osteocytes, then an over-signalling ($Oy=2$) can happen, indicating that for instance a micro-fracture has occurred. Rule \textsc{Oy$_-$} states that $Oy$ can decrease if $Oy>0$ and if $Oy\leq Oc$, meaning that osteocytes' signalling can decrease only when the necessary number of osteoclasts has been recruited. Rule \textsc{Oc$_+$} tells that osteoclasts can proliferate under the following conditions: $Oc < Oy$, meaning that the number of recruited osteoclasts must agree with the intensity of the activity of osteocytes; $Ob\leq 1$, modelling the negative regulation of osteoblasts on osteoclasts; and of course $Oc < 2$, to avoid out-of-range updates. Rule \textsc{Oc$_-$} regulates the death of osteoclasts, which can occur only if $Oc>0$ and if $Oc > Oy$, thus ensuring that osteoclasts cannot decrease before they have been completely recruited. \textsc{Ob$_+$} regulates osteoblasts' proliferation, that must agree with the obvious condition that $Ob < 4$; that $Oy=0$ (no active osteocytes); and that $Ob < 2Oc$ (the number of recruited osteoblasts is proportional to the number of osteoclasts). Finally rule \textsc{Ob$_-$} tells that $Ob$ can decrease only if $Ob > 0$ and if $Ob > Oc$, which makes sure that the formation phase ends only after the resorption phase.

The state machine of the resulting $B$ level is depicted in Fig.~\ref{fig:b-level_br}.
\begin{table}
\centering
\begin{tabular}{rrl}
\hline
\textsc{Init:} &$(0,0,0)\goes{}$&$Oy=Oy+1$\\
\textsc{Oy$_+$:} & $Oy==0 \goes{}$&$Oy=2$\\
\textsc{Oy$_-$:} &$Oy\leq Oc \wedge Oy > 0 \goes{}$&$Oy=Oy-1$\\
\textsc{Oc$_+$:} &$Ob\leq 1 \wedge Oc < Oy \wedge Oc < 2 \goes{}$&$Oc=Oc+1$\\
\textsc{Oc$_-$:} &$Oc > Oy \wedge Oc > 0 \goes{}$&$Oc=Oc-1$\\
\textsc{Ob$_+$:} &$Ob < 2Oc \wedge Oy=0 \wedge Ob < 4 \goes{}$&$Ob=Ob+1$\\
\textsc{Ob$_-$:} &$Ob > Oc \wedge Ob > 0 \goes{}$&$Ob=Ob-1$\\\hline
\end{tabular}
\caption{Guarded rules of the form $\text{\textsc{RuleName:}} \ guard \goes{} update$, determining the transition relation in the $B$ level.}
\label{tbl:BRrules}
\end{table}

\begin{sidewaysfigure}
\centering
\includegraphics[width=\textwidth]{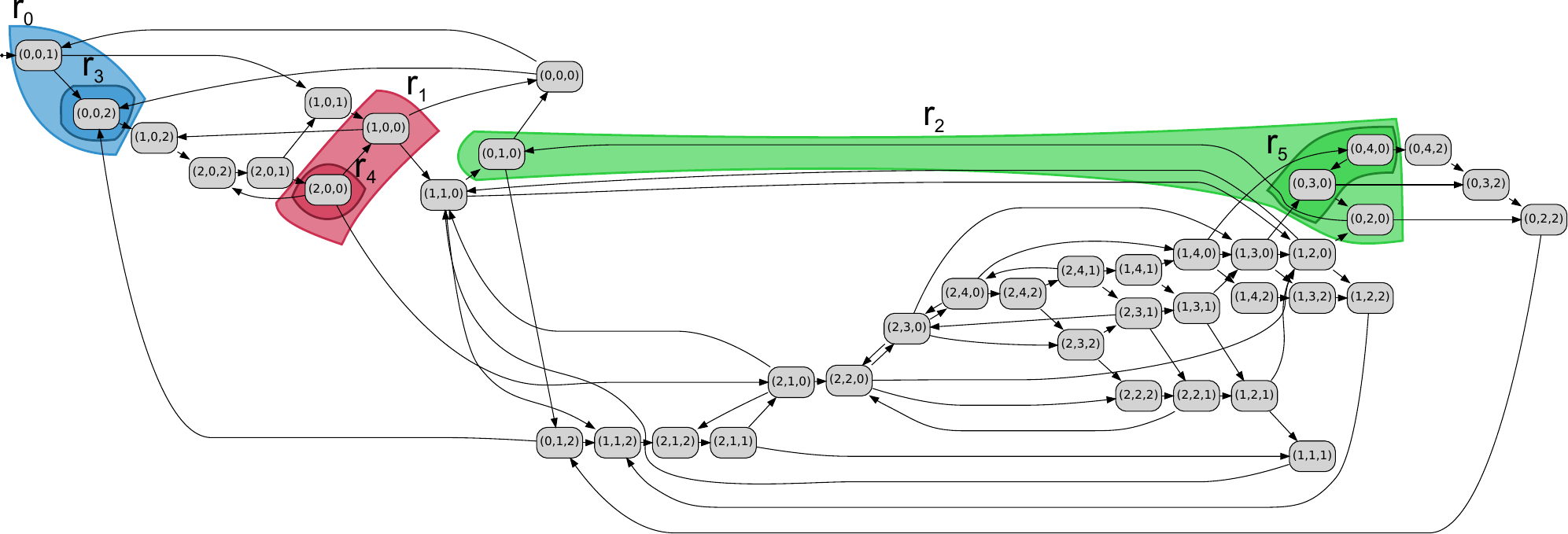}
\caption{The behavioural state machine for the bone remodelling example. Each state is characterized by different evaluation of the variables $(Oc, Ob, Oy)$ (osteoclasts, osteoblasts, osteocytes). Transitions are determined by the rules specified in Table~\ref{tbl:BRrules}. For the sake of clarity, states unreachable from the initial state $(0,0,1)$ have been omitted. Coloured areas are used to represent the states of the $S$ level, which identify stable regions in the $B$ level. The $S$ states considered are: $r_0$ (initiation of remodelling, light blue), $r_1$ (resorption, light red), $r_2$ (formation, light green), $r_3$ (osteocytes' over-signalling, blue), $r_4$ (high resorption, red) and $r_5$ (high formation, green).}
\label{fig:b-level_br}
\end{sidewaysfigure}

\subsubsection{Structural Level}
In this example, the structural level models the different key phases during bone remodelling. Table~\ref{tbl:s-states} lists the considered $S$ states $r$, together with their associated constraints $L(r)$. Additionally, Fig.~\ref{fig:b-level_br} shows how the structural constraints identify different stable regions in the behavioural level.
\begin{itemize}
\item[$r_0:$] it indicates the \textit{initiation} of the remodelling process and requires that osteocytes must be active ($Oy>0$), but that osteoclasts and osteoblasts are not ($Oc==0 \wedge Ob==0$).
\item[$r_1:$] it models the \textit{resorption} phase, occurring when only osteoclasts are active ($Oc>0 \wedge Ob==0 \wedge Oy==0$).
\item[$r_2:$] it describes the \textit{formation} phase, occurring when only osteoblasts are active ($Ob>0 \wedge Oc==0 \wedge Oy==0$).
\item[$r_3:$] it models the occurrence of a ``fault'' in the remodelling system, like an unordinary mechanical stress or a micro-fracture, after which osteocytes' signalling is more prominent ($Oy==2$).
\item[$r_4:$] it describes a high resorption activity  ($Oc>1$).
\item[$r_5:$] it describes a high formation activity  ($Ob>2$).
\end{itemize}

\begin{table}
\centering
\begin{tabular}{rcc}
&$r$&$L(r)$\\\hline
Initiation&$r_0$&$Oy>0 \wedge Oc==0 \wedge Ob==0$\\
Resorption&$r_1$&$Oc>0 \wedge Ob==0 \wedge Oy==0$\\
Formation&$r_2$&$Ob>0 \wedge Oc==0 \wedge Oy==0$\\
Osteocytes' over-expression&$r_3$&$Oy==2 \wedge Oc==0 \wedge Ob==0$\\
High resorption&$r_4$&$Oc>1 \wedge Ob==0 \wedge Oy==0$\\
High formation&$r_5$&$Ob>2 \wedge Oc==0 \wedge Oy==0$\\ \hline
\end{tabular}
\caption{List of $S$ states $r$ and associated labelling function $L(r)$ (constraints) in the bone remodelling example. Each $S$ state models a key phase during the remodelling cycle.}
\label{tbl:s-states}
\end{table}

The two $S$ levels $S_0$ (regular adaptation) and $S_1$ (fault-tolerant adaptation) are illustrated in Fig.~\ref{fig:s-levels_br}. In the following, we denote the set of states and transitions of a structural level $S_i$ with $R(S_i)$ and $\goes{}_S(S_i)$, respectively.
The structural state machine $S_0$ is given by:
\begin{small}
$$S_0 = (\{r_0, r_1, r_2\}, r_0, \mathcal{O}^{\Sigma,A}_{M}, \{ r_0 \goes{Oc>0}_S r_1, r_1 \goes{Ob>0 \wedge Oy==0}_S r_2, r_2 \goes{Ob==0}_S r_0 \}, L ),$$
\end{small}
where $\mathcal{O}^{\Sigma,A}_{M}$ is the above defined observation function; and $L$ is the labelling function as described in Table~\ref{tbl:s-states}. Below we discuss in more detail the transitions of $S_0$.
\begin{itemize}
\item $r_0 \goes{Oc>0}_S r_1.$ During the adaptation between the initiation and the resorption phase, $Oc>0$ must hold, meaning that osteoclasts have to be recruited.
\item $r_1 \goes{Ob>0 \wedge Oy==0}_S r_2.$ The transition invariant tells that during the adaptation between resorption and formation, osteoblasts have to be recruited and osteocytes are no more active.
\item $r_2 \goes{Ob==0}_S r_0.$ This transition requires that the formation phase has to be completed ($Ob==0$), before starting another remodelling cycle.
\end{itemize}
The structural state machine $S_1$ is given by:
$$S_1 = (R(S_0) \ \cup \ \{r_3, r_4, r_5\}, r_0, \goes{}_S(S_1), L), $$
where
\[
\begin{array}{rl}
\goes{}_S(S_1) \ = \ \goes{}_S(S_0) \ \cup &\{ r_2 \goes{Oy>0}_S r_3, r_3 \goes{Oc>0 \wedge Ob==0}_S r_4, \\ 
& r_4 \goes{Ob>0 \wedge Oy==0}_S r_5, r_5 \goes{Oy>0}_S r_3, r_5 \goes{Ob<3}_S r_2\}.
\end{array}
\]
The transitions added in $S_1$ allow us to model the self-adaptation of the bone remodelling system after an unexpected malfunctioning, in response to which a higher remodelling activity (a sort of fallback remodelling cycle) is initiated:
\begin{itemize}
\item $r_2 \goes{Oy>0}_S r_3.$ If after the formation activity, osteocytes have started sending signals ($Oy>0$), then the system adapts to an $S$ state characterized by the over-expression of osteocytes.
\item $r_3 \goes{Oc>0 \wedge Ob==0}_S r_4.$ During the adaptation between osteocytes' over-expression and the high formation activity, osteoclasts have to be recruited ($Oc>0$) and osteoblast must not be active ($Ob==0$).
\item $r_4 \goes{Ob>0 \wedge Oy==0}_S r_5.$  Similarly to the transition between regular resorption and regular formation, during the adaptation between high resorption and high formation, osteoblasts have to be recruited and osteocytes have not to be active.
\item $r_5 \goes{Oy>0}_S r_3.$ A fallback remodelling cycle can take place also after the high formation phase, if $Oy>0$.
\item $r_5 \goes{Ob<3}_S r_2.$ From the high formation phase, the system may adapt to a regular formation activity, under the invariant $Ob<3$.
\end{itemize}

\begin{figure}
\centering
\subfloat[]{\includegraphics[width=0.5\textwidth]{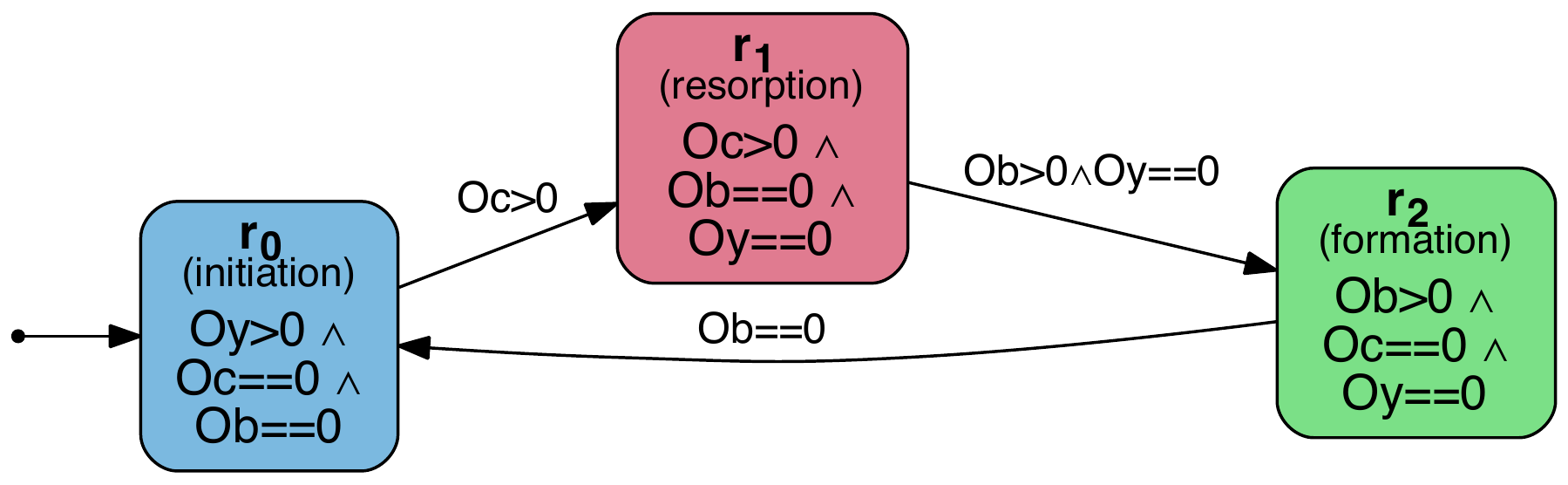}}\\
\subfloat[]{\includegraphics[width=\textwidth]{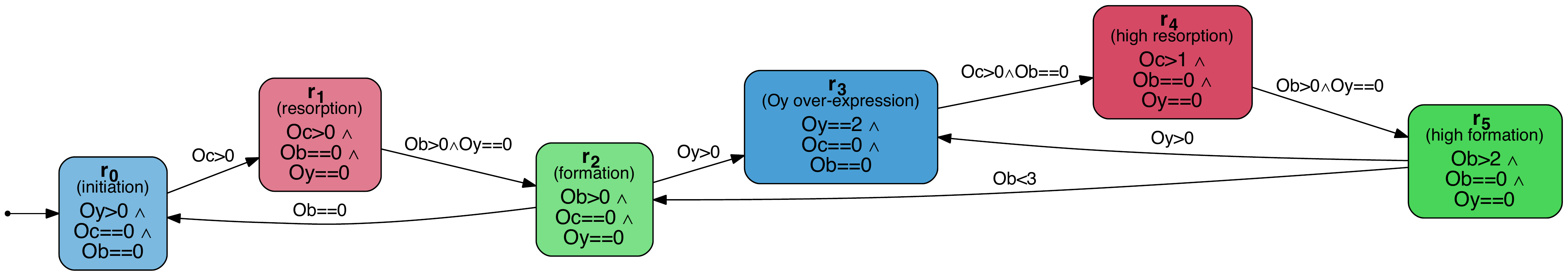}}
\caption{The two different structural levels for the bone remodelling example. $S_0$ (Fig.~\ref{fig:s-levels_br} (a)) models the key phases during a regular remodelling activity ($r_0$ initiation, $r_1$ resorption and $r_2$ formation). $S_1$ (Fig.~\ref{fig:s-levels_br} (b)) extends $S_0$ with a ``fallback'' loop, activated by an over-expression of osteocytes ($r_3$) which in turn triggers a higher resorption activity ($r_4$) and a higher formation ($r_5$).}
\label{fig:s-levels_br}
\end{figure}

\subsection{Flat Semantics of the Bone Remodelling Example}
The flat semantics of $S_0[B]$ and $S_1[B]$ in the bone remodelling example is given in Figure~\ref{fig:sbs_br}.
\begin{figure}
\centering
\subfloat[]{\includegraphics[width=0.15\textwidth]{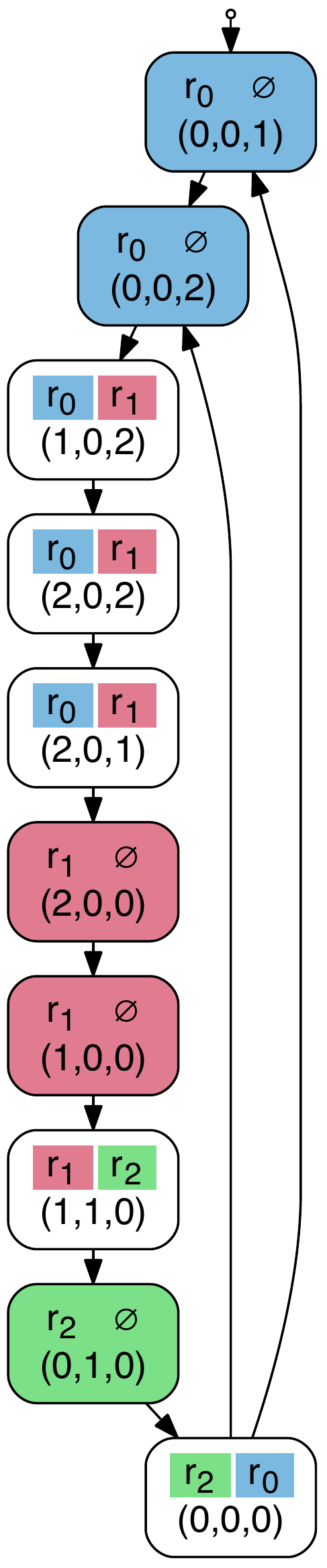}}\hfill
\subfloat[]{\includegraphics[width=0.5\textwidth]{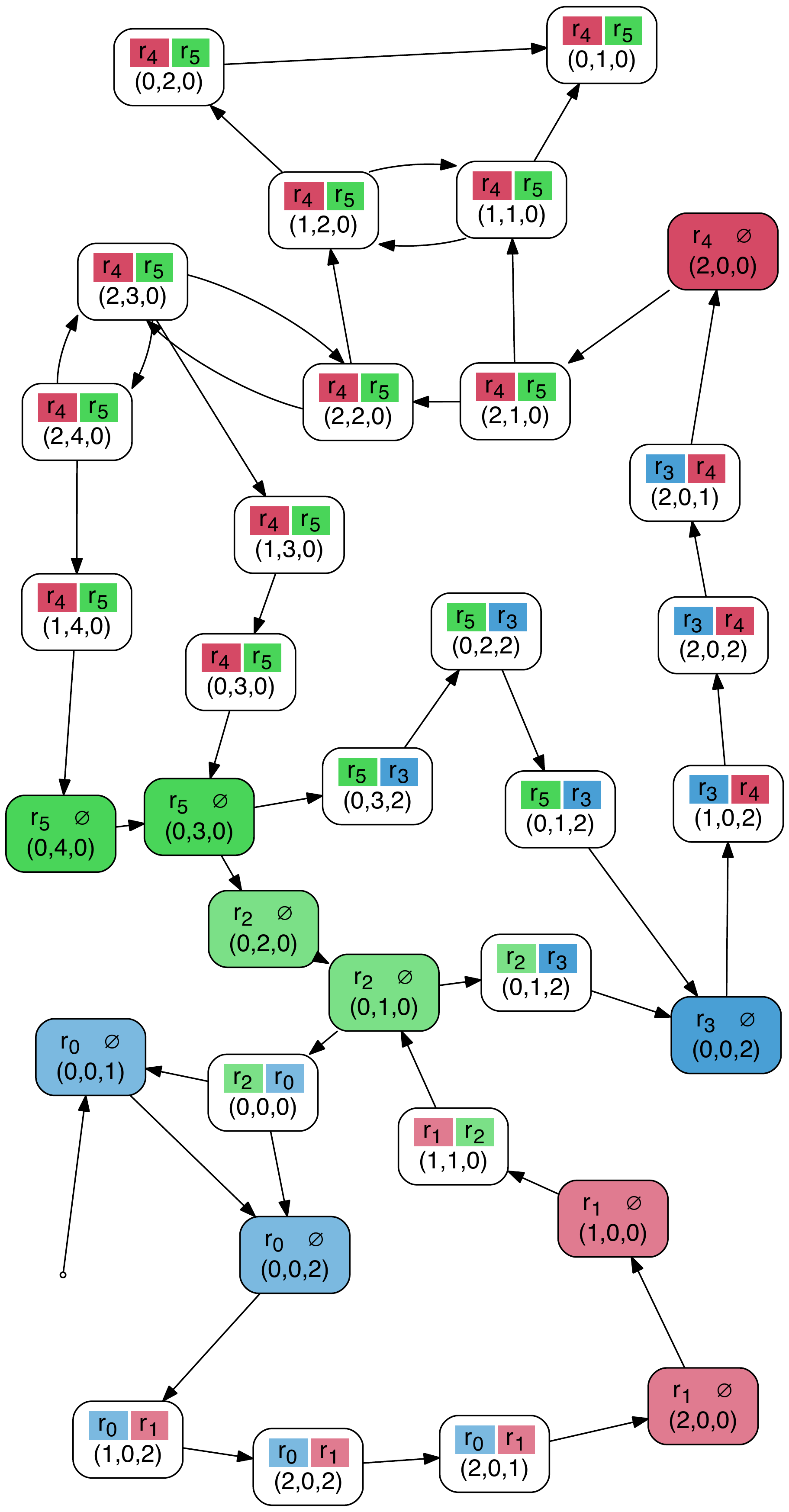}}
\caption{The flat semantics of the two systems $S_0[B]$ (Fig.~\ref{fig:sbs_br} (a)) and $S_1[B]$ (Fig.~\ref{fig:sbs_br} (b)) in the bone remodelling example.
}
\label{fig:sbs_br}
\end{figure}
Firstly, we observe that $\mathcal{F}(S_1[B])$ has a larger state space than the flat semantics of $S_0[B]$, due to the higher number of states and transitions in $S_1$. Since the two systems share the same behavioural level and $S_0$ is a subset of $S_1$, it is possible to notice that also $\mathcal{F}(S_0[B]) \subseteq \mathcal{F}(S_1[B])$.

Similarly to the motion controller model, in $\mathcal{F}(S_0[B])$ every adaptation path leads to a target $S$ state and in $\mathcal{F}(S_1[B])$ there always exists a successful adaptation path, but there are infinite adaptation paths as well and there is a deadlock at the adapting state $((0,1,0), r_4, \{ \cdot, r_5\})$ because all its successors violate the adaptation invariant.

Thus, we can anticipate that the behavioural level $B$ is able to successfully adapt under normal remodelling requirements (structure $S_0$), for \textit{all possible adaptation paths}, or $S_0[B]$ is \textit{strong adaptable}. On the contrary, $B$ is able to successfully adapt under ``fault-tolerant'' remodelling requirements (structure $S_1$), only for \textit{some adaptation paths}, or $S_1[B]$ is \textit{weak adaptable}.

\subsection{Adaptation Relations in the Bone Remodelling Example}
We show that in the bone remodelling case study $S_0[B]$ is strong adaptable and $S_1[B]$ is weak adaptable, but not strong adaptable.

In order to verify that $S_0[B]$ is strong adaptable, we find a strong adaptation relation $\mathcal{R}$ s.t.\ $(q_0, r_0) \in \mathcal{R}$. Similarly to the motion control example, in $\mathcal{F}(S_0[B])$ every state is reachable from the initial state $((0,0,1), r_0, \emptyset)$. Thus, we can consider the relation $\mathcal{R} =\{(q,r) \ | \ (q,r,\emptyset) \in F\}$:
\[
\mathcal{R} = \{ ((0,0,1), r_0), ((0,0,2), r_0), ((2,0,0), r_1), ((1,0,0), r_1), ((0,1,0), r_2) \}.
\]
Note that condition $(i)$ of the definition of strong adaptation is true for every couple in the relation, since $\forall (q,r) \in \mathcal{R}. \ q \models L(r)$; and $\textsc{Progress}(q,r)$ holds for any of such states, because there are no deadlock steady states in the flat semantics. Clearly, $(q_0, r_0) \in \mathcal{R}$. We show that $\mathcal{R}$ is a strong adaptation relation, by checking requirements $(ii)$ and $(iii)$ of the definition of strong adaptation (Sect.~5 of the manuscript) for each element of $\mathcal{R}$.
\begin{small}
\begin{itemize}
\item $((0,0,1), r_0)$.
\begin{itemize}
\item[$(ii)$] $((0,0,1), r_0, \emptyset)\goes{r_0}((0,0,2), r_0, \emptyset)$ and $((0,0,2), r_0) \in \mathcal{R}$
\item[$(iii)$] $((0,0,1), r_0, \emptyset)\goes{r,\psi,r'}\! \! \! \! \! \! \not$
\end{itemize}
\item $((0,0,2), r_0)$.
\begin{itemize}
\item[$(ii)$] $((0,0,2), r_0, \emptyset)\not\goes{r}$
\item[$(iii)$] there is only one adaptation path from $((0,0,2), r_0, \emptyset)$ leading to \\$((2,0,0), r_1, \emptyset)$, and $((2,0,0), r_1), \in \mathcal{R}$.
\end{itemize}
\item $((2,0,0), r_1)$.
\begin{itemize}
\item[$(ii)$] $((2,0,0), r_1, \emptyset)\goes{r_1}((1,0,0), r_1, \emptyset)$ and $((1,0,0), r_1) \in \mathcal{R}$
\item[$(iii)$] $((2,0,0), r_1, \emptyset)\goes{r,\psi,r'}\! \! \! \! \! \! \not$
\end{itemize}
\item $((1,0,0), r_1)$.
\begin{itemize}
\item[$(ii)$] $((1,0,0), r_1, \emptyset)\not\goes{r}$
\item[$(iii)$] there is only one possible adaptation path from $((1,0,0), r_1, \emptyset)$ leading to $((0,1,0), r_2, \emptyset)$, and $((0,1,0), r_2)  \in \mathcal{R}$.
\end{itemize}
\item $((0,1,0), r_2))$.
\begin{itemize}
\item[$(ii)$] $((0,1,0), r_2), \emptyset)\not\goes{r}$
\item[$(iii)$] there are two possible adaptation paths from $((0,1,0), r_2), \emptyset)$ leading to $((0,0,1), r_0, \emptyset)$ and to $((0,0,2), r_0, \emptyset)$, respectively, and\\ $((0,0,1), r_0), ((0,0,2), r_0)  \in \mathcal{R}$.
\end{itemize}
\end{itemize}
\end{small}

We show that $S_1[B]$ is weak adaptable, by the following weak adaptation relation:
\[
\begin{array}{rl}
\mathcal{R} = &\{ ((0,0,1), r_0), ((0,0,2), r_0), ((2,0,0), r_1), ((1,0,0), r_1), ((0,1,0), r_2),\\
& ((0,0,2), r_3), ((2,0,0), r_4),((0,4,0), r_5), ((0,3,0), r_5),((0,2,0), r_2)\}.
\end{array}
\]
Similarly to $S_0[B]$, $(q_0,r_0) \in \mathcal{R}$ and for all $(q,r) \in \mathcal{R}$, $\ q \models L(r)$ and $\textsc{Progress}(q,r)$ both holds. Therefore, we show that $\mathcal{R}$ is a weak adaptation relation, by checking requirements $(ii)$ and $(iii)$ of the weak adaptation definition (Def.~\ref{def:weakadapt}) for each element of $\mathcal{R}$. Since $\mathcal{F}(S_0[B])$ is included in $\mathcal{F}(S_1[B])$, we can omit the test for those states that have been inspected for $S_0[B]$.
\begin{small}
\begin{itemize}
\item $((0,0,2), r_3)$.
\begin{itemize}
\item[$(ii)$] $((0,0,2), r_3, \emptyset)\not\goes{r}$
\item[$(iii)$] $((0,0,2), r_3, \emptyset)\goes{r_3, \psi, r_4}^+((2,0,0), r_4, \emptyset)$ and $((2,0,0), r_4) \in \mathcal{R}$
\end{itemize}
\item $((2,0,0), r_4)$.
\begin{itemize}
\item[$(ii)$] $((2,0,0), r_4, \emptyset)\not\goes{r}$
\item[$(iii)$] $((2,0,0), r_4, \emptyset)\goes{r_4, \psi, r_5}^+((0,4,0), r_5, \emptyset)$ and $((0,4,0), r_5) \in \mathcal{R}$
\end{itemize}
\item $((0,4,0), r_5)$.
\begin{itemize}
\item[$(ii)$] $((0,4,0), r_5, \emptyset)\goes{r_5}((0,3,0), r_5 \emptyset)$ and $((0,3,0), r_5) \in \mathcal{R}$
\item[$(iii)$] $((0,4,0), r_5, \emptyset)\goes{r,\psi,r'}\! \! \! \! \! \! \not$
\end{itemize}
\item $((0,3,0), r_5)$.
\begin{itemize}
\item[$(ii)$] $((0,3,0), r_5, \emptyset)\not\goes{r}$
\item[$(iii)$] $((0,3,0), r_5, \emptyset)\goes{r_5, \psi, r_2}((0,2,0), r_2, \emptyset)$ and $((0,2,0), r_2) \in \mathcal{R}$
\end{itemize}
\item $((0,2,0), r_2)$.
\begin{itemize}
\item[$(ii)$] $((0,2,0), r_2, \emptyset)\goes{r_2}((0,1,0), r_2, \emptyset)$ and $((0,1,0), r_2) \in \mathcal{R}$
\item[$(iii)$] $((0,2,0), r_2, \emptyset)\goes{r,\psi,r'}\! \! \! \! \! \! \not$
\end{itemize}
\end{itemize}
\end{small}
Note that $\mathcal{R}$ is not a strong adaptation because for instance the element $((2,0,0), r_4)$ cannot be in a strong relation, since the adaptation path 
\[
\begin{array}{l}
((2,0,0), r_4, \emptyset) \goes{r_4, \psi, r_5} ((2,1,0), r_4, \{(\phi, r_5)\}) \goes{r_4, \psi, r_5}\\
\goes{r_4, \psi, r_5}((1,1,0), r_4, \{(\phi, r_5)\})\goes{r_4, \psi, r_5}  ((0,1,0), r_4, \{(\phi, r_5)\})
\end{array}
\]
cannot lead to a steady state ($((0,1,0), r_4, \{(\phi, r_5)\})$ is a deadlock adapting state).

\end{document}